\documentclass[journal,twoside,web]{ieeecolor}
\usepackage{amsmath,amssymb,amsfonts}
\usepackage{empheq}
\usepackage{graphicx}
\usepackage[colorlinks=true,linkcolor=magenta,citecolor=blue,urlcolor=cyan,filecolor=red,draft]{hyperref}
\usepackage{textcomp}
\usepackage{xcolor}
\usepackage{graphicx}% Include figure files
\usepackage{times}
\usepackage[T1]{fontenc}
\usepackage{dcolumn}% Align table columns on decimal point
\usepackage{bm}% bold math
\usepackage{colortbl}
\usepackage{generic}
\usepackage{cite}
\usepackage{amsmath,amssymb,amsfonts}
\usepackage{graphicx}
\usepackage{textcomp}
\usepackage{multicol}
\usepackage{multirow}
\usepackage{chngpage}
\usepackage{array}
\usepackage{subfigure}
\usepackage{enumerate}
\usepackage{algpseudocode}  
\usepackage{mathrsfs}
\usepackage{subfigure}
\usepackage{makecell}
\usepackage[most]{tcolorbox}
\usepackage{epsfig} % for postscript graphics files
\newtheorem{theorem}{Theorem}

\newtheorem{lemma}{Lemma}
\newtheorem{remark}{Remark}
\newtheorem{assumption}{Assumption}
\newtheorem{proposition}{Proposition}
\newtheorem{property}{Property}

\newtheorem{mydef}{Definition}

\newcommand{\bM}{{\mathcal{M}}}

\newcommand{\bR}{\mathbb{R}}
\newcommand{\bN}{\mathbb{N}}
\newcommand{\bC}{\mathbb{C}}
\newcommand{\bS}{\mathbb{S}}
\newcommand{\bx}{{\bm{x}}}

\newcommand{\bU}{U^{T_s}}
\newcommand{\hF}{\mathcal{F}}     
\newcommand{\hD}{\mathcal{D}}                 
\newcommand{\bL}{{\widehat{L}}}   
\newcommand{\tL}{{\widetilde{L}}}   
\newcommand{\hL}{\mathcal{L}}

\newcommand{\bIm}{\mathrm{Im}}
\newcommand{\bI}{\mathscr{I}}

\newcommand{\Log}{\mathrm{Log}}
\newcommand{\sspan}{\mathrm{span}}

\newcommand{\hG}{\mathscr{G}}

\def\BibTeX{{\rm B\kern-.05em{\sc i\kern-.025em b}\kern-.08em
		T\kern-.1667em\lower.7ex\hbox{E}\kern-.125emX}}
\begin{document}
	\title{A Sampling Theorem for Exact Identification of 
		Continuous-time Nonlinear Dynamical Systems}
	\author{Zhexuan Zeng, Zuogong Yue, Alexandre Mauroy, Jorge Gon{\c{c}}alves and Ye Yuan
		\thanks{This paragraph of the first footnote will contain the date on 
			which you submitted your paper for review. This work was supported by the National Natural Science Foundation of China under Grant 92167201 }
		\thanks{Zhexuan Zeng, Zuogong Yue and Ye Yuan are with School of Artificial Intelligence and Automation, Huazhong University of Science and Technology, Wuhan, China. }
		\thanks{A. Mauroy is with the Department of Mathematics and the Namur Institute for Complex Systems (naXys), University of Namur, Namur, Belgium.}
		\thanks{J. Gon{\c{c}}alves is with the Luxembourg Centre for Systems Biomedicine, University of Luxembourg, Belvaux, Luxembourg.}
		\thanks{{$^*$}For correspondence, \href{mailto:
				yye@hust.edu.cn}{\tt yye@hust.edu.cn}.}}
	
	\maketitle
	
	\begin{abstract}
		Low sampling frequency challenges the exact identification of the continuous-time (CT) dynamical system from sampled data, even when its model is identifiable. The necessary and sufficient condition is proposed-- which is built from Koopman operator-- to the exact identification of the CT system from sampled data. The condition gives a Nyquist-Shannon-like critical frequency for exact identification of CT nonlinear dynamical systems with Koopman invariant subspaces: 1) it establishes a sufficient condition for a sampling frequency that permits a discretized sequence of samples to discover the underlying system and 2) it also establishes a necessary condition for a sampling frequency that leads to system aliasing that the underlying system is indistinguishable; and 3) the original CT signal does not have to be band-limited as required in the Nyquist-Shannon Theorem. The theoretical criterion has been demonstrated on a number of simulated examples, including linear systems, nonlinear systems with equilibria, and limit cycles. 
	\end{abstract}
	
	\begin{IEEEkeywords}
		System identification, sampling period, Koopman operator.
		%Enter key words or phrases in alphabetical  order, separated by commas. For a list of suggested keywords, send a blank e-mail to keywords@ieee.org or visit \underline {http://www.ieee.org/organizations/pubs/ani\_prod/keywrd98.txt}
	\end{IEEEkeywords}
	
	\section{Introduction}
	\label{sec:introduction}
	\IEEEPARstart{O}{ne} of the central problems in science and engineering is to discover underlying physical laws from observations. Physical laws are invariably described as continuous-time (CT) dynamical systems, while the available measurements are inevitably in discrete-time (DT) \cite{ljung2021modeling}. This leads to a key challenge: when the sampling frequency is not high enough to capture the dynamics, several CT systems may generate the same DT measurements, and as a result, the underlying CT dynamical system can not be exactly identified. Similar to aliasing in signals, aliasing in systems (or we call it system aliasing) is defined to characterize such a many-to-one relationship between the CT systems and DT measurements \cite{Yue2020a}. 
	
	The goal of this paper is to investigate problems of system aliasing of nonlinear dynamical systems, i.e., the condition required to obtain a unique CT dynamical system from the sampled measurements. Based on the equivalence between the vector field and the infinitesimal generator of the Koopman operator\cite{Koopman1931hamiltonian,mezic2005spectral}, our goal is recast in obtaining a unique generator from its Koopman operator \cite{mauroy2019koopman}. Building on results from semigroup theory and operator logarithm, necessary and sufficient conditions to guarantee no system aliasing are proposed when the associated Koopman operator of the dynamical system admits point, continuous or residual spectrum. The results show that system aliasing and accordingly critical sampling period are associated with the imaginary part of the spectrum of the generator of Koopman operators. These results provide a natural extension of known results for linear systems \cite{Yue2020a} and extends the Nyquist-Shannon Sampling Theorem for signals that are sampled from nonlinear dynamical systems with a special Koopman invariant subspace.
	
	There have been efforts in the identification of CT dynamical systems from DT measurements in the past decades; a continuous-time system identification (CONTSID) toolbox has been developed in Matlab \cite{garnier2018contsid}. The  identification approaches of CT dynamical system include indirect methods \cite{sinha1982identification} and direct methods \cite{unbehauen1998review, garnier2008identification, garnier2015direct, subrahmanyam2019identification}. For indirect methods, the CT system is obtained by transformation from the identified DT dynamical system. However, performance of this transformation is highly dependent on the choice of the sampling period \cite{subrahmanyam2019identification}. To overcome this shortcoming, the direct methods use pre-filters to provide smoothed differentiated data and identify the parameters of CT models directly, which outperforms the indirect methods \cite{ljung2003initialisation}. To reconstruct the derivatives of the measurements, the underlying CT signals are assumed to contain no frequency above half of their sampling frequency \cite{garnier2015direct}. However, such a band-limited assumption \cite{shannon1949communication} may not be true for general nonlinear dynamical systems. 
	
	The content of this paper is organized as follows. The problem and concepts related to system aliasing are presented in Section \ref{sec:pre}. The Koopman operator and the spectrum of operators are introduced in Section \ref{sec:Kopm} and \ref{sec:spectrum}, respectively. Then we analyze system aliasing of dynamical systems using eigenvalues (Section \ref{sec:point}) and spectrum (Section \ref{sec:conspectrum}), including continuous / residual spectrum, of the Koopman operator. In the case of Koopman eigenvalues, the notion of Koopman valid eigenspace is introduced in Section \ref{sec:property}, and system aliasing is proposed equivalently with it in Section \ref{sec:Tsi} (Lemma \ref{th1}). Related concepts of generator aliasing are defined to analyze system aliasing in the eigenspace in \ref{sec:inf}. Through semigroup and operator logarithm theorems, the critical sampling period can be expressed in an analytical form (Theorem \ref{th7}) in Section \ref{sec:pri}. In the case of Koopman continuous or residual spectrum, system aliasing is analyzed equivalently in an infinite-dimensional Koopman invariant subspace in Section \ref{sec:infinitesubspace}  and Section \ref{sec:iinf}. The critical sampling period (Theorem \ref{infsampling}) is expressed in Section \ref{sec:ipri}. Some Numerical experiments are provided in Section \ref{sec:num}. Conclusions and possible implications %and limitations 
	of these results are discussed in Section \ref{sec:con}.
	
	\subsection{Notations}\label{sec:not}
	The paper adopts the following notation:
	%$\lambda(\cdot)$ denotes eigenvalues (point spectrum) of operators, $\sigma(\cdot)$ denotes the spectrum of operators, 
	%$\hB$ denotes Banach space, %$\hH$ denotes Hilbert space,
	$\hL(X)$ denotes all linear bounded operators from the Banach space $X$ to itself, $\hD(\cdot)$ denotes the domain of an operator, $\sspan\{\cdot\}$ denotes the linear space spanned by the basis functions, and $\nabla$ denotes the gradient operator. 
	
	\section{Problem Formulation}\label{sec:pre}
	\subsection{A motivating example}\label{sec:moti}
	We first present a motivating example to illustrate an interesting phenomenon that when the sampling frequency is low, the underlying physical systems can not be identified from DT measurements.
	
	Suppose an extendable rod, whose length is $r(t)$ of time $t$, rotating clockwise around its fixed end with a constant angular velocity $\omega$ rad/s (see Fig.~\ref{fig1}). We take a series of photos with the time interval $T_s$ s to record the position of the other end, i.e., the measurements of $x_1(kT_s) = r(kT_s)\cos(k\omega T_s)$ and $x_2(kT_s)=r(kT_s)\sin(k\omega T_s)$, $k=1,2,\ldots$. Fig. \ref{fig1} shows the real trajectories (denoted by blue lines), and the DT measurement at each moment (denoted by black triangles), when the length of the rod is constant or getting longer. In fact, the $x_1(t)$ and $x_2(t)$ in two cases are states of the following linear dynamical systems: \begin{align*}
	\dot{x}_1 &= ax_1+\omega x_2,\\
	\dot{x}_2 &= -\omega x_1+ax_2,
	\end{align*}
	with $a=0$ and $a=0.1$ respectively, which are structural identifiable using identifiability tests in \cite{glover1974parametrizations, yuan2016identification}.
	
However, when the photo-taken frequency is not high enough, it is difficult to infer whether the rod is rotating clockwise or counterclockwise, indeed, there could exist another trajectory with the same DT measurements (denoted by orange dashed lines). Fig. \ref{fig1} shows three photos (the time interval satisfies $\omega T_s = 4\pi/3$) when the length of the rod is constant and getting longer respectively. It implies that we may falsely infer that the rod is rotating counterclockwise (denoted by orange dashed lines), which perfectly matches the photos. From a dynamical system point of view that there are at least two different CT systems that generate the same DT measurements. As a result, without prior information, one can not identify the underlying system using any advanced system identification algorithms. 
	%    , i.e., %not correct dynamical system perfectly matches all the measurements. 
	%    the correct dynamical system cannot be inferred from DT observations due to low sampling frequency. 
	
	To provide a sampling frequency bound to avoid this aliasing phenomenon, Nyquist-Shannon sampling theorem \cite{shannon1949communication} is helpful when the corresponding CT signals are band-limited, such as the case of Fig. \ref{Fig1.sub1}. However, when they are not band-limited, shown Fig. \ref{Fig1.sub2}, the critical sampling frequency to avoid aliasing phenomenon is also important. Therefore, we shall define the concept of system aliasing motivated by this example and discuss the critical sampling frequency to avoid it in the later sections.
	
	\begin{figure}[!t]
		\centering
		\subfigure[The rod $r(t)$ is constant by $\dot{r} = 0$ while rotating with $\omega$ rad/s]{
			\label{Fig1.sub1}
			\centerline{\includegraphics[width=.5\textwidth]{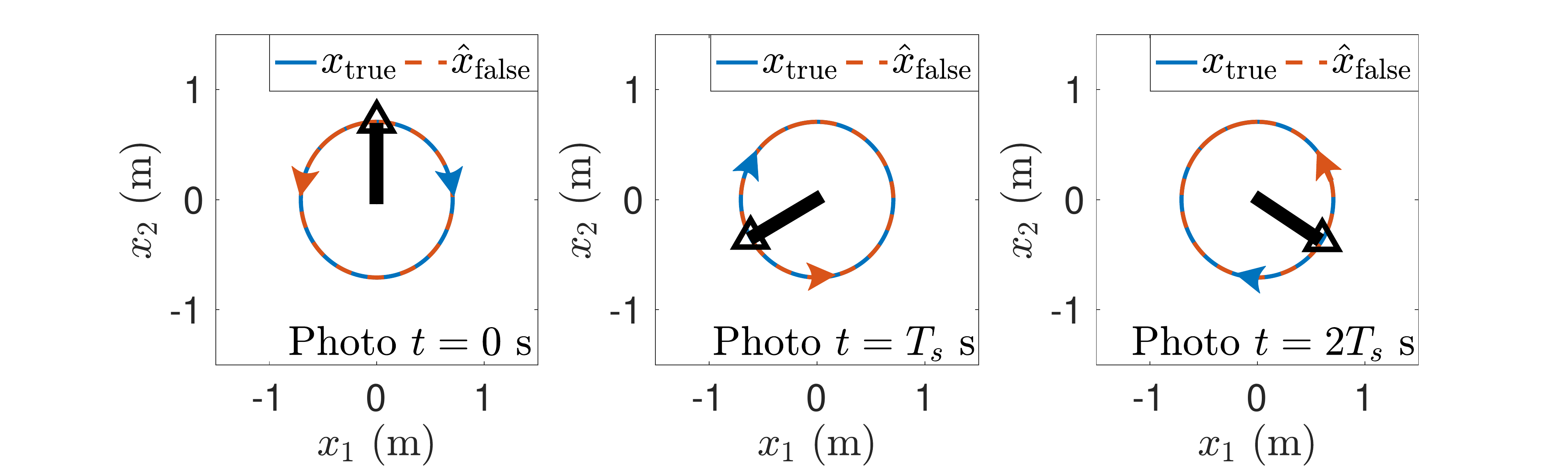}}}
		\subfigure[The rod $r(t)$ is getting longer by $\dot{r} = 0.1r$ while rotating with $\omega$ rad/s]{
			\label{Fig1.sub2}
			\centerline{\includegraphics[width=.5\textwidth]{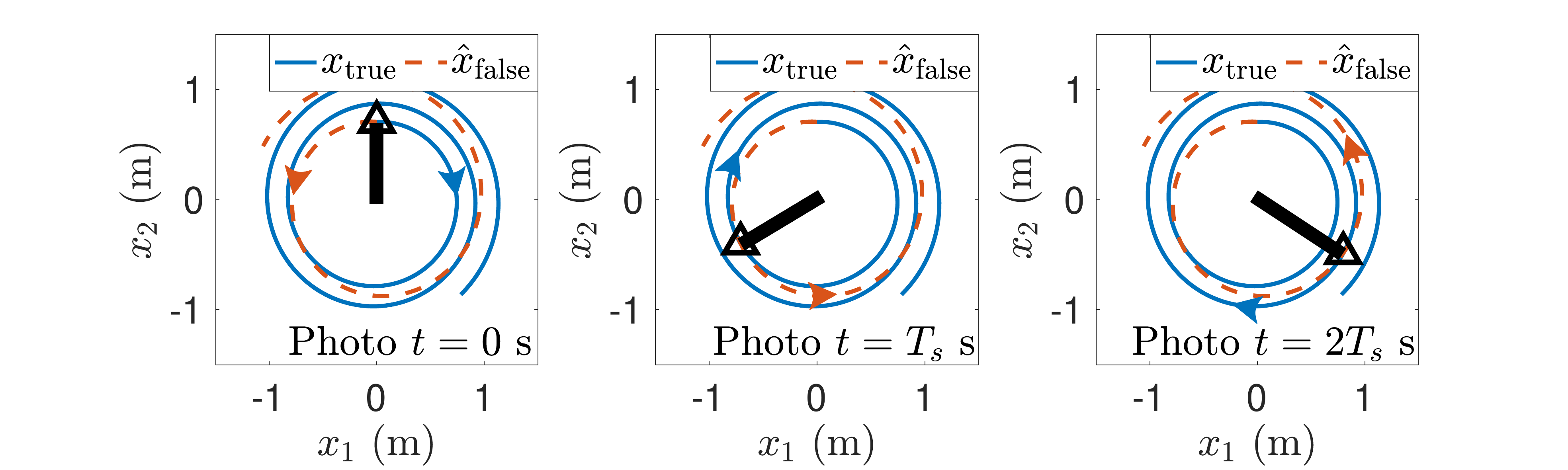}}}
		%		\subfigure[The rod $r(t)$ is getting shorter by $\dot{r} = -0.1r$ while rotating with $\omega$ rad/s]{
		%			\label{Fig1.sub3}
		%			\includegraphics[width=.5\textwidth]{thought_aliasing_shorter_fig1.pdf}}
		\caption{Observation of extendable rod to show system aliasing. The sampling is not high enough to capture the systems of (a) $\dot{r} = 0$ and (b) $\dot{r}=0.1r$. For these dynamical systems, the outputs of true and false dynamical systems are denoted as blue lines and orange dashed lines respectively. The position of rotating end of the rod in each photo is denoted as black triangle.}
		\label{fig1}
	\end{figure}
	
	\subsection{Definitions}\label{sec:def}
	Consider an autonomous nonlinear dynamical system described by the following ordinary differential equations:
	\begin{equation}\label{eq1}
	\begin{aligned}
	\dot{\bx}(t)&=f (\bx)
	\end{aligned}
	\end{equation}
	where $\bx\in \bM\subset \bR^n$, $t\in \bR^+ $ are the state vector and time respectively, with $ \bM$ denoting the state space.
	The vector field $\bx\mapsto f(\bx):\bM \to\bM$ describes the dynamical properties of the states. The flow induced by \eqref{eq1} is : $t\mapsto S^t(\bx_0):R^+ \times \bM\to \bM$, which is a solution of \eqref{eq1} associated with initial states $\bx_0$, i.e., $\bx(t) = S^t(\bx_0)$. In order to study system aliasing specifically, we focus on measurements with equidistant sampling period $T_s$, i.e., time-series $\{\bx(t_k)\}$ with $\bx(t_{k+1})=S^{T_s}(\bx(t_{k}))$, where $t_{k+1}= t_k+T_s$ and $k=1,2,\ldots$. %\ye{talk about k... also time series is $\bx(t_k)$ not this pair..... you have to define it properly}. 
	Assume that the vector field $f(\bx)$ is continuously differentiable %satisfies Lipschitz condition %is continuously differentiable with compact support 
	and the DT flow $S^{T_s}$ can be obtained exactly. The concept of system aliasing is defined as follows.%\ye{$\bx_k,\by_k$?? I thought you only needs $\bx_k$??}
	
	%We considers full-state measurements and focuses on the identifiability in the context of low sampling frequency. We aim at obtaining the condition on sampling frequency to uniquely identify the parameter $\btheta$ with infinite noise-free snapshot pairs $[\bx_k,\by_k]\in\bM^2$, where $\by_k=S^{T_s}(\bx_k), k = 1,2, \ldots$. 
	
	\begin{mydef}[System aliasing] Consider the CT dynamical system \eqref{eq1}. If there exists a vector field $\hat{f}(\bx)\neq f(\bx)$ for some $\bx\in\bM$ that can generate the same discrete flow $S^{T_s}(\bx)$ for all $\bx\in\bM$, where $\hat{f}\in \bI_f$ and $\bI_f$ is the aliasing space of the vector field $f$, we call $\hat{f}$ a system alias of $f$ with respect to $\bI_f$. %\ye{do you mean that $\hat{f}(x)\neq f(x),~\forall x$}
	\end{mydef}
	%\ye{do you need $\forall x_0$......?}
	
	\begin{remark}
		This definition is a generalization of the definition of system aliasing in the linear case \cite{Yue2020a}. By this general definition, no system aliasing means that the vector field can be recovered from discrete flow $S^{T_s}$ with the equidistant sampled data. 
	\end{remark}
	
	Accordingly, the concepts of critical sampling period / frequency bound are defined as follows. 
	\begin{mydef}[Critical sampling period / frequency bound] 
		The critical sampling period of the CT dynamical system \eqref{eq1} is $T_\gamma$ s if there is no system aliasing of the correct dynamical system for all $t< T_\gamma$ s. Equivalently, the lower bound of sampling frequency is $2\pi/T_\gamma$ rad/s.
	\end{mydef}
	%\ye{be careful! I have changed the notation $T_\gamma$ to denote the critical period, and $T_s$ for generic period}
	\begin{remark}
	The lower bound of sampling frequency can provide a criterion to collect data for system identification, which guarantees that the vector field can be exactly identified. The identification results of CT dynamical system would be unreliable if the sampling frequency is lower than its bound, since there exits another vector field that can also generate these sampled data. 
	\end{remark}
	
	Our goal is to investigate the condition to guarantee that there is no system alias for dynamical systems, and obtain the critical sampling frequency $T_\gamma$ for nonlinear dynamical systems. 
	
	\section{Analysis of system aliasing using Koopman operator}
	
	The analysis of system aliasing exploits the equivalent descriptions of nonlinear systems and infinite-dimensional linear systems through the Koopman operator. The main idea to analyze system aliasing is illustrated in Fig. \ref{fig2}. First, the nonlinear dynamical system is lifted to a Koopman invariant space. Secondly, system aliasing in state space is expressed in terms of generator aliasing of the Koopman operator in this space. Finally, operator logarithm theorems are used to derive the necessary and sufficient condition of system aliasing. 
	
	\begin{figure}[thpb]
		\centering
		%\framebox{
		\includegraphics[width=0.5\textwidth]{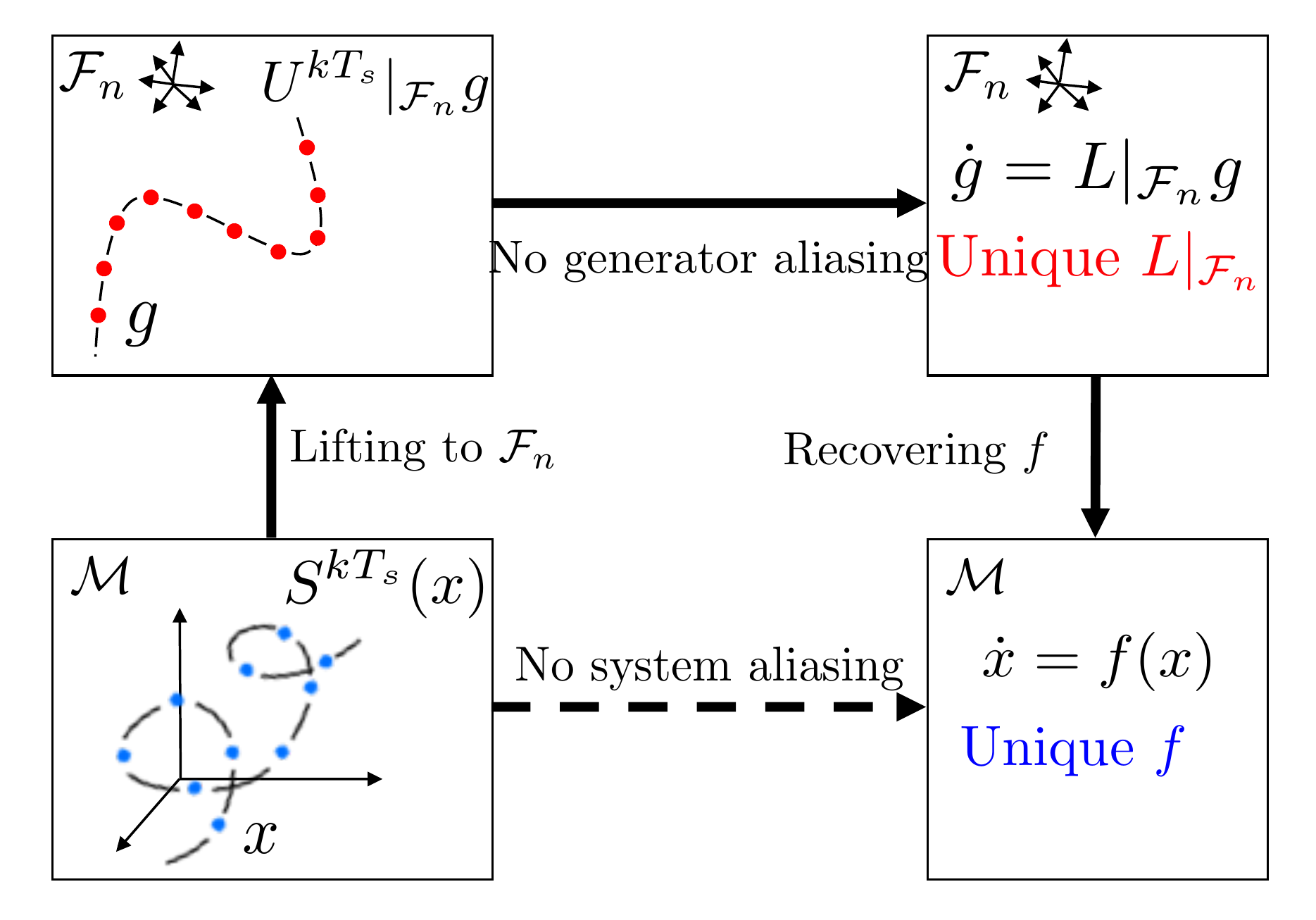}
		%}
		%\includegraphics[scale=1.0]{figurefile}
		\caption{The steps to analyze system aliasing of nonlinear dynamical systems (For illustration purpose, we only consider nonlinear systems with finite-dimensional Koopman eigenspace $\hF_n$). The analysis of system aliasing consists of three steps: (1) lifting the discrete flow $S^{T_s}$ to $\hF_n$, (2) obtaining the unique generator from the Koopman operator $\bU|_{\hF_n}$, and (3) recovering the vector field $f$.}
		\label{fig2}
	\end{figure}
	
	%----------------------------------------------------------------------------
	
	\subsection{Operator-theoretic approach to dynamical systems}\label{sec:Kopm}
	%----------------------------------------------------------------------------
	The flow generated by \eqref{eq1} is described by $S^t$ in the state space $\bM$. Equivalently, the system can be described by the Koopman semigroup in the (Banach) space $\hF$. The semigroup of the Koopman operator\cite{Koopman1931hamiltonian} describes the evolution of the observable functions along the trajectories:
	\begin{equation}\label{eq4}
	U^{t} g=g\circ S^{t} , \,\, g\in\hF,
	\end{equation} where $g:\bM\to\bC$ is an observable function. The Koopman operator is a linear operator, i.e., let the observable functions $g_1,g_2\in \hF$,
	\begin{equation*}
	U^t(\alpha_1 g_1+\alpha_2g_2)=\alpha_1U^t g_1+\alpha_2U^t g_2, \ \alpha_1,\alpha_2\in \bR.
	\end{equation*}
	The dynamical interpretation of $U^t$ is as follows: the linear transformation of the observable, i.e., $U^t g(\bx)$, has the value of observable function $g$ at $S^t(\bx)$. Specifically, if $g(\bx)=\bx_j$, where $\bx_j\ (1\le j\le n)$ is the $j$th component of $\bx$, $U^t g(\bx(0))$ has the value that the initial state $\bx_j(0)$ flows after the lapse of the time $t$. %Therefore, Koopman operator $U^t$, which is an infinite-dimensional and linear operator, interprets the flow $S^t$ in operator framework.
	
	The infinitesimal generator of the Koopman operator can be interpreted as the derivative of $U^t$, %provides to study dynamical properties of states, i.e. $f$ in \eqref{eq3}, and this is the existence of i, 
	which is defined by:
	\begin{equation}\label{eq5}
	L g=\lim_{t \to0^+}\frac{1}{t}(U^t-I)g, \ g\in \hD(L).
	\end{equation}
	%for observable functions $g_1,g_2\in \hD(L)$. \ye{what are $g_1, g_2$?}
	When the generator $L$ acts on the full observable functions, i.e., $g(\bx)=\bx$, one obtains the vector field $L g = f$. This connection between the generator $L$ and $f(x)$ stems from the formula:
	\begin{equation*}
	\lim_{t \to0^+}\frac{1}{t}(S^t(x_0)-x_0)=f(x_0).
	\end{equation*}
	If the observable function $g$ is continuously differentiable with compact support, there is a general relationship between the dynamics $f$ and the generator $L$ of the corresponding Koopman operator as follows\cite[Section 7.6]{lasota2013chaos}:
	\begin{equation}\label{eq6}
	L g = f \cdot \nabla g,\  g\in \hD(L).
	\end{equation}
	%which provides tools to study nonlinear dynamical system as if they were linear:
	The infinitesimal generator $L$ is also a linear operator. Exploiting \eqref{eq4}-\eqref{eq6}, the analysis of nonlinear systems by Koopman operators can be interpreted as an extension problem of linear systems, from finite-dimensional to infinite-dimensional ones. In particular, the state space $\bM$ is extended to an observable space $\hF$ so that the nonlinear vector field $f$ becomes embedded in a linear operator $L$. The extended linear system of observable functions $g\in \hF$ can be described by 
	\begin{equation}\label{eq7}
	\begin{aligned}
	\dot{g}&=L g,\\
	U^{T_s}g&=g\circ S^{T_s},
	\end{aligned}
	\end{equation}where $\dot{g}$ denotes the ``derivative'' of the observable function $\partial(g\circ S^t)/\partial{t}$\cite{mauroy2020koopman}. 
	%----------------------------------------------------------------------------
	\subsection{Spectrum of the Koopman operator}\label{sec:spectrum}
	The spectrum $\sigma(\cdot)$ of operators is composed of three parts: point spectrum $\sigma_p(\cdot)$, residual spectrum $\sigma_r(\cdot)$, and continuous spectrum $\sigma_c(\cdot)$. In this section, these three different types of spectrum are introduced.
	\subsubsection{Point spectrum} 
	The point spectrum $\sigma_p(B)$ of the operator $B\in \hL(X)$ is the set of value $\lambda$ such that $B - \lambda I$ is not injective, where $I$ is the identity operator.
	When the Koopman operator admits the point spectrum, the eigenvalue and associated eigenfunction of the Koopman operator are defined as follows.
	\begin{mydef}[Koopman eigenvalue and eigenfunctions]
		An eigenvalue of the Koopman operator is the value $\lambda_i\in\bC$ that:
		$$
		U^t\phi_i = e^{\lambda_it}\phi_i,
		$$
		where $\phi_i\in\hF\neq0$ is the associated Koopman eigenfunction.%\ye{define $\sigma_p(U^t)$}
	\end{mydef}
	
	It follows \eqref{eq4} that the Koopman eigenfunction $\phi_i$ is also an eigenfunction of the generator $L$, that is,
	$$
	L\phi_i = \lambda_i \phi_i.
	$$ The value $\lambda_i$ is the eigenvalue of the generator $L$ and belongs to its point spectrum $\sigma_p(L)$.
	
	The set of eigenfunctions and eigenvalues of the Koopman operator can be characterized by the following property.
	\begin{property}[Multiplication of eigenfunctions \cite{mauroy2016global}]\label{property1} Consider the eigenfunctions $\phi_1\in\hF$ and $\phi_2\in \hF$ of the Koopman operator $U^t:\hF\to\hF$ with associated eigenvalues $\lambda_1$ and $\lambda_2$. When $\phi_1^{c_1}\phi_2^{c_2}$ belongs to the space $\hF$, where  $c_1,c_2\ge0$, the observable function $\phi_1^{c_1}\phi_2^{c_2}$ is also an eigenfunction of $U^t$ with the associated eigenvalue $c_1\lambda_1 + c_2\lambda_2$.
	\end{property} 
	
	\subsubsection{Continuous spectrum}
	The Koopman operator may admit continuous or residual spectrum, which is typically associated with chaotic systems. The continuous spectrum of the operator $B\in\hL(X)$ is defined as follows.
	\begin{mydef}[Continuous spectrum\\ {\cite[p. 12]{mauroy2020koopman}}]
		The continuous spectrum $\sigma_c(B)$ of the operator $B\in\hL(X)$ is the set of values $\lambda \in \bC $ such that the operator $B-\lambda I$ is injective and has a dense image, but not surjective, where $I$ is the identity operator.
	\end{mydef}
	
	\subsubsection{Residual spectrum}
	Here is the definition of residual spectrum of the operator $B\in\hL(X)$.
	\begin{mydef}[Residual spectrum\\ {\cite[p. 12]{mauroy2020koopman}}]
		The residual spectrum $\sigma_r(B)$ of the operator $B\in \hL(X)$ is the set of values $\lambda \in \bC $ such that the operator $B-\lambda I$ is injective, not surjective, and does not have a dense image, where $I$ is the identity operator.
	\end{mydef}
	
	%	\ye{reposition}When the Koopman semigroup $U^t$ is uniformly continuous, i.e., the generator $L$ is bounded, the spectral mapping theorem is as follows.
	%	\begin{equation*}
	%	\sigma(U^t) = e^{t \sigma(L)}:=\{e^{\lambda t}:\lambda\in\sigma(L) \}.
	%	\end{equation*}
	
	Next, we shall consider system aliasing for nonlinear systems with these three different spectrum of associated Koopman operators.

	\section{Sampling theorem for systems with Koopman point spectrum}\label{sec:point}
	When the Koopman operator $U^t:\hF\to\hF$ admits point spectrum, there exists infinite Koopman eigenfunctions $\{\phi_k\}_{k=1}^\infty$. In this section, we consider the dynamical systems whose Koopman operator has $n$ valid eigenfunctions (Definition \ref{def1}). With such $n$ Koopman valid eigenfunctions, which satisfy some requirements, the vector field of the $n$-dimensional dynamical system can be recovered (Proposition \ref{property3}). The sampling theorem to avoid system aliasing of these dynamical systems is proposed with a Koopman valid eigenspace, which is spanned by $n$ valid eigenfunctions (Theorem \ref{th7}). For the dynamical system whose Koopman valid eigenfunctions are lacking or there does not exists point spectrum, the system aliasing and the sampling theorem will be considered in Section \ref{sec:conspectrum}.
	\subsection{Properties of the Koopman valid eigenspace}\label{sec:property}
	%	For most dynamical systems with hyperbolic attractors, the Koopman operator defined on well-chosen spaces of observable functions only admits a point spectrum %the continuous and residual parts of the spectrum are empty with well-chosen space of observable functions
	%	\cite{mauroy2016global}. 
	Consider a set of Koopman eigenfunctions whose gradients can be linearly independent vectors for $\bx\in\bM$, which we called valid eigenfunctions. We define the Koopman valid eigenspace $\hF_n$ as follows.
	\begin{mydef}[Koopman valid eigenspace $\hF_n$]\label{def1}
		The Koopman valid eigenspace $\hF_n$ is spanned by $n$ eigenfunctions of the Koopman operator, i.e., $\hF_n = \sspan\{\phi_1,\ldots,\phi_n\}$, where $n$ is the dimension of the state space and the gradients of these eigenfunctions $\{\nabla\phi_i\}_{i=1}^n$ should be linearly independent vectors for all $\bx\in \bM$, i.e., the matrix
		$A_{ij} = \partial \phi_i/\partial x_j$ is invertible for all $\bx\in \bM$.%\ye{the latter part is a property not a definition}
		
		% should satisfy the following requirements:
		%\begin{itemize}
		%\item Eigenfunctions are independent, i.e. the only one solution of integers $c_k\ge 0$ such that $\phi_i=\sum_{k=1}^{n}c_k\phi_k$ is $\{c_{k=i} = 1, c_{k\neq i}=0\}$.
		%\item Eigenvalues are independent, i.e. the only one solution of integers $c_k\ge 0$ such that  $\phi_i=\prod_{k=1}^n\phi_k^{c_k}$ is $\{c_{k=i} = 1, c_{k\neq i}=0\}$.
		%\end{itemize}
	\end{mydef}
	
	\begin{remark}[Principal eigenvalues]
		The Koopman eigenfunctions that associated to the principal eigenvalues can span the valid eigenspace in Definition \ref{def1}. For $n$-dimensional dynamical systems with different transient and asymptotic behaviors, such as dynamical systems with equilibrium, limit cycle and quasi-periodic attractors, the spectrum of associated Koopman operator is a lattice generated by the linear combination of integers of $n$ principal eigenvalues, provided that the operator is defined on well-chosen spaces of observable functions \cite{mezic2020spectrum}. The gradients of the associated eigenfunctions are linearly independent vectors for all $\bx$ in the basin of the attraction. Therefore, the $n$ principal eigenfunctions associated with $n$ distinct eigenvalues span a proper eigenspace according to Definition \ref{def1}. However, if an eigenvalue is generated by the linear combination of eigenvalues associated to the other $n-1$ eigenfunctions that already belong to $\hF_n$ (see Property \ref{property1}), the independence condition on the gradient of the eigenfunctions may not be satisfied. When a principal eigenvalue is repeated which results in the lack of eigenfunctions, it needs to consider the continuous or residual spectrum. 
	\end{remark}
	
	The following proposition provides a necessary and sufficient condition to recover the vector field $f$ from the generator. 
	
	\begin{proposition}[Conditions to recover $f$]\label{property3}
		Consider the Koopman invariant subspace $\hF = \sspan\{\phi_j\}_{j=1}^N$ and the generator $L|_{\hF}$, which is restricted to the space $\hF$. The vector field $f$ in \eqref{eq1} can be recovered from $L|_\hF$ if and only if there are at least $n$ Koopman eigenfunctions $\phi_i\in\hF$ with gradients $\{\nabla \phi_i\}_{i=1}^n$ being linearly independent vectors for all $\bx\in \bM$.
	\end{proposition}
	\begin{proof}
		%{\proof
		\emph{Sufficiency}. Consider the states $\bx\in\bM\subset \bR^n$ and the vector field $f$. If at least $n$ Koopman eigenfunctions with linearly independent vectors $\{\nabla\phi_i\}_{i=1}^n$ for all $\bx\in \bM$ are given, the vector field $f = [f_1,\ldots,f_n]$ can be obtained through \eqref{eq6} and $n$ eigenvalue equations
		\begin{equation*}
		\sum_{i=1}^n f_i \partial{\phi_i}/\partial{x_i} = \lambda_i \phi_i,
		\end{equation*} where $i=1,\ldots,n$.
		Indeed, we have $A f(\bx)= [\lambda_1 \phi_1(\bx) \ldots \lambda_n \phi_n(\bx)]^T$,
with the invertible matrix $A_{ij} = \partial \phi_i/\partial x_j(\bx)$.
		
		\emph{Necessity}. Consider the Koopman invariant subspace $\hF = \sspan \{\phi_1,\ldots,\phi_N \}$, where $\{\phi_j\}_{j=1}^N$ are eigenfunctions of $L|_{\hF}$. The vector field $f=[f_1,\ldots,f_n]$ can be recovered from $L|_{\hF}$ by the eigenvalue equations $f \cdot \nabla \phi_j = \lambda_j \phi_j,j=1,\ldots,N$, if the matrix $A_{ij} = \partial \phi_i/\partial x_j(\bx)$ is full rank, i.e., there are at least $n$ linearly independent vectors $\{\nabla\phi_j\}_{j=1}^n$ for all $\bx\in\bM$.	
		% and $n\le j$. Consider the eigenfunction $\phi = \phi_1^{c_1}\phi_2^{c_2}$ generated by Property \ref{property1}, where $c_1,c_2$ are positive integers. Then $\nabla \phi$ is linearly dependent with $\{\nabla\phi_i\}_{i=3}^n$ for $\bx^d\in\bM$ if $\{\nabla\phi_i\}_{i=1}^n$ are linearly dependent vectors for $\bx^d\in\bM$. combined with the maximum number of principal eigenfunctions is $n$, there are at least $n$ eigenfunctions $\phi_i\in\hF$ with gradients $\{\nabla\phi_i\}_{i=1}^n$ linearly independent for all $\bx\in\bM$.
		%}
	\end{proof}
	
	By Proposition \ref{property3}, the eigenspace $\hF_n$ (in Definition \ref{def1}) is a core and valid space to recover the vector space. Denote by $U^t|_{\hF_n}$ and $L|_{\hF_n}$ the restriction of the Koopman operator and its generator to $\hF_n$. Restricted to this finite-dimensional Koopman invariant space $\hF_n$, the vector field can be recovered from the generator $L|_{\hF_n}$. %The continuous part of the spectrum is empty and only $n$ eigenvalues need to be considered. 
	The boundedness of Koopman operator $U^t|_{\hF_n}$ and its generator $L|_{\hF_n}$ can be characterized as follows. 
	\begin{property}[Boundedness of operators]\label{property2}
		Koopman operator $U^t|_{\hF_n}$ and its generator $L|_{\hF_n}$ are both bounded. 
	\end{property}
	\begin{proof}
		%{\proof    
		Consider the observable function $g\in\hF_n, g=\sum_{i=1}^n a_i\phi_i$. The evolution of $g$ is governed by the Koopman operator:$$
		U^t|_{\hF_n} g = U^t|_{\hF_n} (\sum_{i=1}^n a_i\phi_i)=\sum_{i=1}^n a_i\lambda_i \phi_i\in\hF_n.$$
		Therefore, the eigenspace $\hF_n$ is an $n$-dimensional invariant Koopman subspace. Since a finite dimensional linear operator can be represented by a matrix, $U^t|_{\hF_n}$ and $L|_{\hF_n}$ are both bounded. 
		%}
	\end{proof}

	%----------------------------------------------------------------------------
	%----------------------------------------------------------------------------	
	%\section{A Necessary and Sufficient Condition for system aliasing}\label{sec:condition}
	
	\subsection{Koopman valid eigenspace and system aliasing}\label{sec:Tsi}
	Using the operator-theoretic description of nonlinear systems, system aliasing of dynamical systems with Koopman valid eigenfunctions can be equivalently described through the uniqueness of the generator $L|_{\hF_n}$ obtained from $U^{T_s}|_{\hF_n}$, which is described as follows.
	\begin{lemma}[System aliasing with $\hF_n$]\label{th1}%one to one relation between L and f
		There is no system alias of the dynamical system described by \eqref{eq1} when the sampling period is $T_s$, if and only if there exists a generator $L|_{\hF_n}$ that can be obtained from $U^{T_s}|_{\hF_n}$ uniquely, where $\hF_n$ is a Koopman valid eigenspace defined in Definition \ref{def1}.
	\end{lemma}
	\begin{proof}
		%{\proof
		The Koopman operator $\bU$ is accurate by \eqref{eq4} and discrete flow $S^{T_s}$.
		
		\emph{Sufficiency}. It follows Proposition \ref{property3} that the vector field $f$ could be recovered uniquely from the generator $L|_{\hF_n}$ for any one of the eigenspace $\hF_n$, which is restricted in any $\hF_n$. Therefore, there is no system alias of the dynamical system if there exists a generator $L|_{\hF_n}$ that can be obtained from the Koopman operator $\bU|_{\hF_n}$ uniquely. 
		
		\emph{Necessity}. There is a one-to-one relationship between the vector field $f$ and the generator $L$ through $L = f\cdot\nabla$, where the generator $L$ could be defined on any Koopman invariant space. %The converse method is used to verify the one-to-one relationship between the vector field $f$ and the generator $L$, where generator $L$ could be defined on any Koopman invariant space. Assume that $L_1=L_2$ while $f_1\neq f_2$. If $L_1=L_2$, it follows from \eqref{eq6} that $f_1\cdot \nabla g= f_2\cdot \nabla g, \ g\in \hF.$ This results in $f_1=f_2$ and is not consistent with the assumption. Then assume that $f_1=f_2$ while $L_1\neq L_2$. If $f_1=f_2$, \eqref{eq6} implies that $L_1 =f_1\cdot \nabla , L_2 =f_2\cdot \nabla$ and $L_1=L_2$. This is not consistent with the assumption neither. 
		Therefore, the generator $L$ is uniquely defined on every Koopman invariant space when the vector field $f$ is unique. It follows that $L|_{\hF_n}$ can be obtained uniquely from the Koopman operator $U^{T_s}$ if $f$ can be uniquely obtained from $S^{T_s}$, i.e., there is no system alias.
		%	}
	\end{proof}
	
	It follows Lemma \ref{th1} that we can focus on the condition for obtaining a unique generator $L|_{\hF_n}$ from the Koopman operator $\bU|_{\hF_n}$.
	
	\subsection{Generator aliasing in eigenspace}\label{sec:inf}
	By Property \ref{property2}, the Koopman operator $\bU|_{\hF_n}$ and its generator $L|_{\hF_n}$ are both bounded. In this case, the exponential relationship between Koopman semigroup $U^t|_{\hF_n}$ and its generator $L|_{\hF_n}$ are defined as follows:  
	\begin{lemma}[Exponential formulas{\cite[Section 11.8]{hille1996functional}}]\label{le3}
		The semigroup $U^t|_{\hF_n}\in \hL(X)$ is an exponential function of its generator $L|_{\hF_n}$, i.e., $U^t|_{\hF_n}=\exp(L|_{\hF_n}t)$. The exponential of the bounded operator is defined by: \begin{equation*}
		\exp(L|_{\hF_n}t)=\sum_{n=0}^{\infty}\frac{t^n {L|_{\hF_n}}^n}{n!}.
		\end{equation*}
	\end{lemma} 
	
	By Lemma \ref{le3}, the following equation between the Koopman operator $U^{T_s}|_{\hF_n}$ and its infinitesimal generator $L|_{\hF_n}$ holds:
	\begin{equation}\label{eq9}
	U^{T_s}|_{\hF_n} = \exp(L|_{\hF_n}T_s).
	\end{equation}However, there are infinite solutions of $L|_{\hF_n}$ based on \eqref{eq9}, which are described as follows.
	
	\begin{lemma}[Many-valued function{\cite[Section 5.4]{hille1996functional}}]\label{le4}
		Let $B\in \hL(X)$, $I$ denotes the identity operator and $I_k$ denotes idempotent, i.e., $\exp(2\pi i I_k)=I$,
		\begin{equation*}
		\log(\exp B)=B+2\pi i \sum_{k}n_k I_k,
		\end{equation*}
		where the $n_k$ are integers.
	\end{lemma}
	
	While the infinitesimal generator of Koopman semigroup $\{U^t|_{\hF_n},t\ge 0\}$ is unique by the definition in \eqref{eq5}, there are generator aliases that generate the same discrete $U^t|_{\hF_n}$ when we fix $t=T_s$ by \eqref{eq9} and Lemma \ref{le4}. Unfortunately, we cannot identify which one is the ground truth $L|_{\hF_n}$ without additional prior information, which is described as generator aliasing. The concepts of aliasing space of generator and generator aliasing are defined as follows.
	
	\begin{mydef} [Aliasing space of generator $\bI_L$] \label{def6}
		The aliasing space of $L|_{\hF_n}$ is defined as $\bI_L :=\{\tL\in \hL(\hF_n): \max|\bIm(\sigma_p(\tL))|\le\max|\bIm(\sigma_p(L|_{\hF_n}))|\}$.
	\end{mydef}
	
	\begin{remark}
		The choice of aliasing space $\bI_{L}$ is motivated by the properties of the operator logarithm described in Lemma \ref{le4}. The principle is to restrict the feasible set $\bI_L$ by the maximum value of imaginary parts of the eigenvalue of $L|_{\hF_n}$. This maximum value confines the elements of infinitesimal generator that have complex dynamical behaviors, as \cite{Yue2020a}. Since the generator $L|_{\hF_n}$ is one-to-one related to the vector field $f$, the aliasing space of vector field $\bI_f$ depends on the aliasing space of generator $\bI_{L}$.
	\end{remark}
	
	\begin{mydef}[Generator aliasing]\label{def7}
		Given $L|_{\hF_n}\in \hL(\hF_n)$ and $T_s\in \bR^+$.
		The operator $\bL$ is a generator alias of $L|_{\hF_n}$ with respect to $\bI_L$ if there exists $\bL\in \bI_L$ such that $\exp(L|_{\hF_n}T_s)=\exp(\tL T_s)$ with $\bL\ne L|_{\hF_n}$. 
	\end{mydef}
	
	Accordingly, the alias set is $G (L|_{\hF_n},T_s,\bI_{L}):=\{\tL\in \bI_L: \exp(L|_{\hF_n}T_s)=\exp(\tL T_s)\}$, which is the set of all possible generator alias. There is no generator aliasing when $G(L|_{\hF_n},T_s,\bI_{L})=\{L|_{\hF_n}\}$. Note that generator aliasing depends on the correct generator $L|_{\hF_n}$ and the sampling period $T_s$. By Lemma \ref{th1}, generator aliasing is an expression of system aliasing in the framework of the Koopman operator. Hence, the condition of system aliasing could be expressed in terms of $L|_{\hF_n}$ and $T_s$. The lower bound to avoid sampling frequency could be expressed in terms of $L|_{\hF_n}$, which is related to the vector field $f$. %, which is related to the complexity of nonlinear dynamical properties. , which embedded with the vector field $f$,
	%----------------------------------------------------------------------------	
	\subsection{Main theorems}\label{sec:pri}
	%----------------------------------------------------------------------------
	
	In this section, we study the condition that guarantees the uniqueness of $L|_{\hF_n}$ from $U^{T_s}|_{\hF_n}$, i.e., no generator aliasing. To discuss the problem of what $(L|_{\hF_n}, T_s)$ satisfies $G(L|_{\hF_n},T_s,\bI_{L})=\{L|_{\hF_n}\}$, we introduce the essential theorem of linear operator logarithm.
	
	\begin{lemma}[Principal Logarithm{\cite[Thm 2]{krabbe1956logarithm}}]\label{le5}
		Consider the operator $Y\in \hL(X)$ whose logarithm is well defined. There is a unique principal logarithm of $Y$, denoted as $B=\Log(Y)$, where the spectrum of $B$ lie in the strip $\hG(B)=\{z \in \bC:-\pi <\bIm(z)< \pi \}$, i.e., $\sigma(B)\subset \hG(B)$. %\ye{should it be a definition?}
	\end{lemma}
	
	It always holds that $L|_{\hF_n}\in G(L|_{\hF_n},T_s,\bI_{L})$. No generator aliasing requires $L|_{\hF_n}$ to be the only element in the set $G(L|_{\hF_n},T_s,\bI_{L})$. By Lemma \ref{le5}, the condition of no generator aliasing is that $T_s L|_{\hF_n}$ be the principal logarithm of the Koopman operator $U^{T_s}|_{\hF_n}$, i.e., $L|_{\hF_n} = \Log(U^{T_s}|_{\hF_n})/T_s$, where $\Log(\cdot)$ denotes the principal logarithm. Hence, a necessary and sufficient condition to avoid system aliasing of dynamical systems with Koopman valid eigenfunctions is described as follows.
	\begin{proposition}\label{th6} 
		There is no system alias of the dynamical system \eqref{eq1}, with the Koopman operator admitting the point spectrum, if and only if $\sigma_p(L|_{\hF_n})\subset \hG(T_s)$ for at least one Koopman valid eigenspace $\hF_n$, where $\hG(T_s):=\{z \in \bC:-\pi/T_s<\bIm(z)<\pi/T_s\}$
	\end{proposition}
	\begin{proof}
		%{\proof
		\emph{Sufficiency}. It follows from Property \ref{property2} that the generator $L|_{\hF_n}$ is bounded and it can be estimated as follows based on \eqref{eq9}:
		\begin{equation*}
		\bL=\frac{1}{T_s} \Log(\bU|_{\hF_n}).
		\end{equation*} By Lemma \ref{le5}, there exist a $T_s L|_{\hF_n}$ being the principal logarithm of $\bU|_{\hF_n}$, where the generator $L|_{\hF_n}$ satisfies \begin{equation*}
		\sigma(L|_{\hF_n}) = \sigma_p(L|_{\hF_n})\subset \{z \in \bC:-\pi/T_s<\bIm(z)<\pi/T_s\}.
		\end{equation*} Therefore, unique generator $L|_{\hF_n}$ can be obtained from $\bU|_{\hF_n}$. By Lemma \ref{th1}, there is no system alias of the dynamical system.
		
		\emph{Necessity}. By Lemma \ref{th1}, there exists a generator $L|_{\hF_n}$ such that it can be obtained uniquely from $\bU|_{\hF_n}$ if there is no system alias of the dynamical system. Lemma \ref{le4} and Lemma \ref{le5} imply that $T_s L|_{\hF_n}$ should be the principal logarithm of $\bU|_{\hF_n}$ to guarantee the uniqueness of $L|_{\hF_n}$. Therefore, the condition of $\sigma(L|_{\hF_n}) = \sigma_p(L|_{\hF_n})\subset \hG(T_s)$ holds in at least one $\hF_n$, where $\hG(T_s):=\{z \in \bC:-\pi/T_s<\bIm(z)<\pi/T_s\}$.%}
	\end{proof}
	
	Proposition \ref{th6} implies that the aliasing of the dynamical system depends on the sampling period $T_s$ and the dynamical system itself. The dynamical system can be identified by sampled data if eigenvalues of $T_s L|_{\hF_n}$ fall into the strip $\{z \in \bC:-\pi<\bIm(z)<\pi\}$. Since the eigenvalues of $T_s L|_{\hF_n}$ depend on the choice of the functional space, the condition of $\sigma_p(L|_{\hF_n})\subset \hG(T_s)$ may not hold in every Koopman valid eigenspace $\hF_n$. Fortunately, if any $\hF_n$ such that eigenvalues of $L|_{\hF_n}$ fall into the strip $\hG(T_s)$, there is no system alias of the dynamical system.

	Based on Proposition \ref{th6}, obtaining $L|_{\hF_n}$ from $\bU|_{\hF_n}$ requires a small enough period $T_s$ such that the set $\sigma_p(L|_{\hF_n})$ falls into the strip $\hG(T_s)$. If it holds, the principal logarithm of $U^{T_s}|_{\hF_n}$ refers to the generator $L|_{\hF_n}$. The critical sampling period are described as follows.
	
	\begin{theorem}[The critical sampling period $T_\gamma$]\label{th7}
		For the $n$-dimensional dynamical system \eqref{eq1} that has $n$ valid eigenfunctions of the Koopman operator, If time-series $\{\bx(t_k) \}$%$[\bx(t_k),\bx(t_{k+1})] $ 
		are equidistant sampled with sampling period $T_s$, to uniquely identify the vector field $f$ from the corresponding discrete flow $S^{T_s}$, the sampling frequency $w$ (rad/s) must satisfy \begin{equation*}
		w>2\min_{\hF_n}\{\max|\bIm(\sigma_p(L|_{\hF_n}))|\}.
		\end{equation*}
		Equivalently, the sampling period should satisfy $T_s<T_\gamma$, where the critical sampling period $T_\gamma ~($i.e., $2\pi/w)$ can be computed by 
		\begin{equation*}
		\boxed{T_\gamma = \max_{\hF_n}\left\{\frac{\pi}{\max|\bIm(\sigma_p(L|_{\hF_n}))|}\right\}}
		\end{equation*}
	\end{theorem}
	\vspace{1mm}
	\begin{proof}
		%{\proof	
		The result immediately follows by Proposition \ref{th6}.
		%}
	\end{proof}
	
	Theorem \ref{th7} provides the critical sampling period and the lower bound of sampling frequency with respect to system aliasing of dynamical systems with the Koopman point spectrum. Since the condition of the system aliasing needs to hold in only one valid eigenspace $\hF_n$, $T_\gamma$ requires the ``best'' $\hF_n$, i.e., $\hF_n^\gamma$ defined below to allow the maximum of the sampling period bound: 	
	\begin{equation*}
	\hF_n^\gamma \triangleq \arg\max_{\hF_n}\left\{\frac{\pi}{\max|\bIm(\sigma_p(L|_{\hF_n}))|}\right\}.
	\end{equation*}
	
	\section{Sampling theorem for systems with Koopman continuous or residual spectrum}\label{sec:conspectrum}
	When the Koopman operator $U^t:\hF\to\hF$ dose not admit point spectrum, or there is a lack of valid eigenfunctions, we can not find a finite-dimensional Koopman valid eigenspace in Definition \ref{def1}. In this case, the point spectrum may not be empty, but it is not enough to recover the vector field $f$. Therefore, the whole spectrum $\sigma(L)$ of the generator, including continuous and residual spectrum, should be considered. In this section, system aliasing is analyzed with the existence of a special (infinite-dimensional) Koopman invariant subspace $\hF_e$, and the sampling theorem is proposed with the spectrum of the generator restricted to this space $\hF_e$.
	
	\subsection{Infinite-dimensional Koopman invariant subspace and system aliasing}\label{sec:infinitesubspace}
	
	Here is an assumption of (infinite-dimensional) Koopman invariant space $\hF_e$. Denote $U^t|_{\hF_e}$ and $L|_{\hF_e}$ as the restriction of the Koopman operator and its generator to $\hF_e$. 
	\begin{assumption}\label{invariant}
		Assume that there are some infinite-dimensional Koopman invariant (Banach) subspaces $\{\hF^1_e,\ldots,\hF_e^p \}$, such that the generator $L|_{\hF^k_e}$ restricted to $\hF_e^k$ is bounded, where $k=1,\ldots, p$.
	\end{assumption}
	
	\begin{remark}
		Assumption \ref{invariant} implies that the evolution of the state vector $\bx(t)$ can be captured by the Koopman operator $U^t|_{\hF_e^k}$, which describes the evolution of observable functions whose gradients are limited. Whether the assumption holds or not is determined by the property of the flow $S^t$ with well-defined observable space $\hF_e^k$ and the choice of the norm. Consider the space $\hat{\hF}\subset\hF$ such that the generator restricted to it $L|_{\hat{\hF}}$ is bounded. To make the space $\hat{\hF}$ invariant under the Koopman operator, it should satisfy $$\sup_{\|g \circ S^t\|=1} \|f\cdot \nabla (g\circ S^t)\|= \sup_{\|g \circ S^t\|=1} \|f\cdot J(S^t) \nabla g\|<\infty,$$ where $J(S^t)$ is the Jacobian matrix of $S^t$, i.e., $[J(S^t)]_{ij} = \partial S^t_j/\partial x_i$, and $\nabla g = [\partial g/\partial x_1,\ldots,\partial g/\partial x_n]$. 
		
		For example, consider a dynamical system whose Koopman operator admits a continuous spectrum \cite{mezic2020spectrum}, which is described as follows.
	\begin{align*}
	\dot{r} = 0,\\
	\dot{\theta} = r,
	\end{align*}where the action-angle variables $(r, \theta)\in [a,b]\times \bS^1$ and $[a,b]\in \bR^+$. The Koopman semigroup $U^t$ which acts on the space $L^2([a,b]\times \bS^1)$ has an eigenvalue $0$ %with associated eigenfunction $g(r,\theta) = r$ 
	and continuous spectrum. There is a special Koopman invariant space $\hF_e$ that $\{g\in\hF_e:~g(r,\theta) = h(r)e^{i\theta}, $ where $ h(r) \in L^2[a,b] \}$, such that the generator $L|_{\hF_e}$ is bounded.
		However, this boundedness assumption may be limited for some dynamical systems, and we shall leave the unbounded case of the generator for future work.
	\end{remark}
	
	%	The infinite-dimensional Koopman invariant subspace $\hF_e^k$ is assumed to contain the valid observable functions, which is described below. \ye{delete this assumption}
	%	\begin{assumption}\label{independent}
	%		Assume that there is at least $n$ observable functions $\{g_j\}_{j=1}^n, g_j\in \hF_e^k$ such that $\{\nabla g_j\}_{i=1}^n$ are linearly independent vectors for all $\bx\in \bM$.
	%	\end{assumption}
	%	\begin{remark}
	%		Assumption \ref{independent} should hold for Koopman invariant subspaces $\hF_e^k,k = 1,\ldots,p$, since they are infinite-dimensional and could generate infinite number of $\{\nabla g_j\}_{i=1}^n$ for all $\bx\in\bM$.
	%	\end{remark}
	%	
	The following proposition provides a necessary and sufficient condition to recover the vector field $f$ from the generator $L|_{\hF_e^k}$.
	
	\begin{proposition}[Conditions to recover $f$]\label{proposition2}
		Consider an infinite-dimensional Koopman invariant subspace $\hF_e^k$ and the corresponding generator $L|_{\hF_e^k}$. The vector field $f$ in \eqref{eq1} can be recovered from $L|_{\hF_e^k}$ if and only if there are at least $n$ observable functions $g_j\in \hF_e^k$ with gradients $\{\nabla g_j\}_{i=1}^n$ being linearly independent vectors for all $\bx\in \bM$.
	\end{proposition}
	\begin{proof}
		%{\proof
		\emph{Sufficiency}. Consider the states $\bx\in\bM\subset \bR^n$ and the vector field $f = [f_1,\ldots,f_n]$. If at least $n$ Koopman eigenfunctions with linearly independent vectors $\{\nabla g_j\}_{j=1}^n$ for all $\bx\in \bM$ are given, the vector field $f$ can be obtained through \eqref{eq6} and the eigenvalue equation $f \cdot \nabla g_j = L|_{\hF_e^k} g_j$. Indeed, we have $A f(\bx)= [L|_{\hF_e^k} g_1(\bx) \ldots L|_{\hF_e^k} g_n(\bx)]^T$, with the invertible matrix $A_{ij} = \partial g_j/\partial x_j(\bx)$.
		
		\emph{Necessity}. Consider the infinite-dimensional Koopman invariant subspace $\hF_e^k$ and the observable function $g_j\in\hF_e^k, j=1,2,\ldots$. The vector field $f=[f_1,\ldots,f_n]$ can be recovered from the eigenvalue equations $f \cdot \nabla g_j = L|_{\hF_e^k} g_j$, if the matrix $A_{ij} = \partial g_j/\partial x_j(\bx)$ is full rank, i.e., there are at least $n$ linearly independent vectors $\{\nabla g_j\}_{j=1}^n$ for all $\bx\in\bM$.	
		% and $n\le j$. Consider the eigenfunction $\phi = \phi_1^{c_1}\phi_2^{c_2}$ generated by Property \ref{property1}, where $c_1,c_2$ are positive integers. Then $\nabla \phi$ is linearly dependent with $\{\nabla\phi_i\}_{i=3}^n$ for $\bx^d\in\bM$ if $\{\nabla\phi_i\}_{i=1}^n$ are linearly dependent vectors for $\bx^d\in\bM$. combined with the maximum number of principal eigenfunctions is $n$, there are at least $n$ eigenfunctions $\phi_i\in\hF$ with gradients $\{\nabla\phi_i\}_{i=1}^n$ linearly independent for all $\bx\in\bM$.
		%}
	\end{proof}
	
	According to Proposition \ref{proposition2}, the generator restricted to the Koopman invariant spaces $\hF_e^k, k=1,\ldots, p$ (in Assumption \ref{invariant}) are able to recover the vector field since they are infinite-dimensional and the observable function $g:\bM\to \bC$ could generate infinite number of $\{\nabla g_j\}_{i=1}^n$ for all $\bx\in\bM$. Therefore, uniqueness of the vector field can be equivalently described through uniqueness of the generator $L|_{\hF_e^k}$ obtained from $U^{T_s}|_{\hF_e^k}$ using the operator-theoretic perspective, which is described as follows.
	%Using the operator-theoretic description of nonlinear systems, system aliasing of dynamical systems with Koopman continuous or residual spectrum can be equivalently described through uniqueness of the generator $L|_{\hF_n}$ obtained from $U^{T_s}|_{\hF_n}$. This is described as follows.

	\begin{lemma}[System aliasing with $\hF_e^k$]\label{le6}%one to one relation between L and f
		When the finite-dimensional Koopman valid eigenspace $\hF_n$ (Definition \ref{def1}) does not exist, there is no system alias of the dynamical system \eqref{eq1} when the sampling period is $T_s$, if and only if there exists a generator $L|_{\hF_e^k}$ that can be obtained from $U^{T_s}|_{\hF_e^k}$ uniquely, where $\hF_e^k$ is a Koopman invariant space in Assumption \ref{invariant}.
	\end{lemma}
	\begin{proof}
		The proof is similar to the proof of Lemma \ref{th1}, which considers the infinite-dimensional Koopman invariant subspaces $\hF_e^k$ in Assumption \ref{invariant}, instead of the finite-dimensional Koopman valid eigenspace $\hF_n$.
	\end{proof}
	%	\begin{proof}
	%		%{\proof
	%		The Koopman operator $U^{T_s}|_{\hF_e^k}$ is accurate by \eqref{eq4} and the discrete flow $S^{T_s}$.
	%		
	%		\emph{Sufficiency}. It follows from Proposition \ref{proposition2} and Assumption \ref{independent} that the vector field $f$ could be recovered uniquely from the generator $L|_{\hF_e^k}$ for any one of the infinite-dimensional Koopman invariant subspace $\hF_e^k$ which meets Assumption \ref{invariant}-\ref{independent}. Therefore, there is no system alias of the dynamical system if there exists a generator $L|_{\hF_e^k}$ that can be obtained from the Koopman operator $\bU|_{\hF_e^k}$ uniquely. 
	%		
	%		\emph{Necessity}. There is a one-to-one relationship between the vector field $f$ and the generator $L$ through $L = f\cdot\nabla$, where the generator $L$ could be restricted to infinite-dimensional Koopman invariant space $\hF_e^k$. 
	%		Therefore, the generator $L|_{\hF_e^k}$ is uniquely when the vector field $f$ is unique. It follows that $L|_{\hF_e^k}$ is unique if $f$ can be uniquely obtained from $S^{T_s}$, i.e., there is no system alias.
	%		%	}
	%	\end{proof}
	
	It follows from Lemma \ref{le6} that we can focus on the condition for obtaining a unique generator $L|_{\hF_e^k}$ from the Koopman operator $\bU|_{\hF_e^k}$.
	
	\subsection{Generator aliasing in infinite-dimensional Koopman invariant space}\label{sec:iinf}
	In this section, we focus on the condition to obtain the generator $L|_{\hF_e^k}$ uniquely from its Koopman operator $\bU|_{\hF_e^k}$. 
	
	Since the Koopman semigroup $U^t
	|_{\hF_e^k}$ is usually bounded and the generator $L|_{\hF_e^k}$ is also bounded by Assumption \ref{invariant}, the exponential relationship between the Koopman semigroup $U^t|_{\hF_e^k}$ and the generator $L|_{\hF_e^k}$ defined in Lemma \ref{le3} is valid, i.e.,
	\begin{align*}
	U^t|_{\hF_e^k}=\exp(L|_{\hF_e^k}t)= \sum_{n=0}^{\infty}\frac{t^n {L|_{\hF_e^k}}^n}{n!}.
	\end{align*}

	When the time $t = T_s$, it follows that:
	\begin{equation}\label{eq10}
	U^{T_s}|_{\hF_e^k} = \exp(L|_{\hF_e^k}T_s).
	\end{equation} However, there are also several (in fact infinitely many) solutions of $L|_{\hF_e^k}$ based on \eqref{eq9}, which are described in Lemma \ref{le4}.
	
	Although the generator of Koopman semigroup $\{U^t|_{\hF_e^k},t\ge 0\}$ is unique by the definition in \eqref{eq5}, there are generator aliases that generate the same discrete $U^t|_{\hF_e^k}$ when we fix $t=T_s$ by \eqref{eq10} and Lemma \ref{le4}. Like the case of the generator restricted to finite-dimensional eigenspace, we cannot identify which one is the ground truth $L|_{\hF_e^k}$ without additional prior information, which is described as generator aliasing. The concepts of aliasing space associated to the generator restricted to the infinite-dimensional Koopman invariant subspace and generator aliasing are defined as follows.
	
	\begin{mydef} [Aliasing space of generator $\bI_{L}^e$] The aliasing space of $L|_{\hF_e^k}$ is defined as $\bI_L^e :=\{\tL_e\in \hL(\hF_e^k): \max|\bIm(\sigma(\tL_e))|\le\max|\bIm(\sigma(L|_{\hF_e^k}))|\}.$
	\end{mydef}
	
	\begin{mydef}[Generator aliasing]
		Consider the generator $L|_{\hF_e^k}\in \hL(\hF_e^k)$ and $T_s\in \bR^+$.
		The operator $\bL$ is a generator alias of $L|_{\hF_e^k}$ with respect to $\bI_L^e$ if there exists $\bL\in \bI_L^e$ such that $\exp(L|_{\hF_e^k}T_s)=\exp(\tL_e T_s)$ with $\bL\ne L|_{\hF_e^k}$. 
	\end{mydef}
	
	\begin{remark}
		This aliasing space $\bI_{L}^e$ and generator aliasing of $L|_{\hF_e^k}$ are similar to Definition \ref{def6}-\ref{def7}, which are also motivated by the properties of the operator logarithm described by Lemma \ref{le4}. The difference is that the domain of the operator in $\bI_{L}^e$ changes from finite-dimensional space to infinite-dimensional space.%The principle is to restrict the feasible set $\bI_L$ by the maximum value of imaginary parts of the eigenvalue of $L|_{\hF_e^k}$. This maximum value confines the elements of infinitesimal generator that have complex dynamical behaviors, as \cite{Yue2020a}. Since the generator $L|_{\hF_e^k}$ is one-to-one related to the vector field $f$, the aliasing space of vector field $\bI_f$ depends on the aliasing space of generator $\bI_{L}^I$.
	\end{remark}
	
	Accordingly, the alias set is $G (L|_{\hF_e^k},T_s,\bI_{L}^e):=\{\tL_e\in \bI_L^e: \exp(L|_{\hF_e^k}T_s)=\exp(\tL_e T_s)\}$, which is the set of all possible generator alias. There is no generator aliasing when $G(L|_{\hF_e^k},T_s,\bI_{L}^e)=\{L|_{\hF_e^k}\}$. Note that the generator aliasing depends on the correct generator $L|_{\hF_e^k}$ and the sampling period $T_s$. %By Lemma \ref{le6}, generator aliasing is an expression of system aliasing in the framework of the Koopman operator. 
	The condition to avoid system aliasing can be expressed in terms of $L|_{\hF_e^k}$ and $T_s$. The lower bound of sampling frequency can be expressed in terms of $L|_{\hF_e^k}$, which corresponds to  $f$. %, which is related to the complexity of nonlinear dynamical properties. , which embedded with the vector field $f$,
	%----------------------------------------------------------------------------	
	\subsection{Main theorems}\label{sec:ipri}
	%----------------------------------------------------------------------------
	
	In this section, we study the condition that guarantees the uniqueness of $L|_{\hF_e^k}$ from $U^{T_s}|_{\hF_e^k}$, i.e., $L|_{\hF_e^k}$ is the only element in the set $G(L|_{\hF_e^k},T_s,\bI_{L}^e)$. 
	
	Under Assumption \ref{invariant}, the Koopman operator $U^{T_s}|_{\hF_e^k}$ is restricted to an infinite-dimensional Banach space $\hF_e^k$. Therefore, Lemma \ref{le5} is also valid for the unique logarithm of $U^{T_s}|_{\hF_e^k}$, and there is a unique principal logarithm of $U^{T_s}|_{\hF_e^k}$, denoted as $B = \Log(U^{T_s}|_{\hF_e^k})$, where $\sigma(B)\subset \hG(B)$. %\ye{should it be a definition?}
	%	It always holds that $L|_{\hF_e^k}\in \bE(L|_{\hF_e^k},T_s,\bI_{L}^I)$. No generator aliasing requires $L|_{\hF_e^k}$ to be the only element in the set $\bE(L|_{\hF_e^k},T_s,\bI_{L}^I)$. 
	By Lemma \ref{le5}, the condition of no generator aliasing is that $T_s L|_{\hF_e^k}$ be the principal logarithm of the Koopman operator $U^{T_s}|_{\hF_e^k}$, i.e., $L|_{\hF_e^k} = \Log(U^{T_s}|_{\hF_e^k})/T_s$. Hence, a necessary and sufficient condition to avoid system aliasing of dynamical systems with continuous / residual spectrum is described as follows.
	\begin{proposition}\label{infmain1}
		There is no system alias of the dynamical system \eqref{eq1} with Koopman continuous / residual spectrum if and only if $\sigma(L|_{\hF_e^k})\subset \hG(T_s)$ for at least one infinite-dimensional Koopman invariant subspace $\hF_e^k$ satisfying Assumption \ref{invariant}, where $\hG(T_s):=\{z \in \bC:-\pi/T_s<\bIm(z)<\pi/T_s\}$
	\end{proposition}
	\begin{proof}
		This proof is based on Lemma \ref{le5}, which is similar to the proof of Proposition \ref{th6}.
	\end{proof}
	%	\begin{proof}
	%		%{\proof
	%		\emph{Sufficiency}. It follows from Property \ref{property2} that the generator $L|_{\hF_e^k}$ is bounded and it can be estimated as follows based on \eqref{eq9}:
	%		\begin{equation*}
	%		\bL=\frac{1}{T_s} \Log(\bU|_{\hF_e^k}).
	%		\end{equation*} By Lemma \ref{le5}, there exist a $T_s L|_{\hF_e^k}$ being the principal logarithm of $\bU|_{\hF_e^k}$, where the generator $L|_{\hF_e^k}$ satisfies \begin{equation*}
	%		\sigma(L|_{\hF_e^k})\subset \{z \in \bC:-\pi/T_s<\bIm(z)<\pi/T_s\}.
	%		\end{equation*} Therefore, unique generator $L|_{\hF_e^k}$ can be obtained from $\bU|_{\hF_e^k}$. By Lemma \ref{th1}, there is no system alias of the dynamical system.
	%		
	%		\emph{Necessity}. By Lemma \ref{th1}, there exists a generator $L|_{\hF_e^k}$ such that it can be obtained uniquely from $\bU|_{\hF_e^k}$ if there is no system alias of the dynamical system. Lemma \ref{le4} and Lemma \ref{le5} imply that $T_s L|_{\hF_e^k}$ should be the principal logarithm of $\bU|_{\hF_e^k}$ to guarantee the uniqueness of $L|_{\hF_e^k}$. Therefore, the condition of $\sigma(L|_{\hF_e^k})\subset \hG(T_s)$ holds in at least one $\hF_e^k$, where $\hG(T_s):=\{z \in \bC:-\pi/T_s<\bIm(z)<\pi/T_s\}$%}
	%	\end{proof}
	%	
	
	Proposition \ref{infmain1} can be seen as a general expression of Proposition \ref{th6}, which considers the spectrum of the Koopman operator, instead of only the point spectrum. %However, compared to the case of point spectrum and the valid Koopman eigenspace, the continuous / residual spectrum is difficult to analyze. 
	It implies that the aliasing of the dynamical system, which has continuous / residual spectrum of the Koopman operator, also depends on the sampling period $T_s$ and the spectrum of the associated generator, which is related to the dynamical system itself. The dynamical system can be identified by sampled data if the spectrum of $T_s L|_{\hF_e^k}$ falls into the strip $\{z \in \bC:-\pi<\bIm(z)<\pi\}$. Since the spectrum of $T_s L|_{\hF_e^k}$ depends on the choice of the functional space, the condition of $\sigma(L|_{\hF_e^k})\subset \hG(T_s)$ may not hold in every Koopman invariant subspace $\hF_e^k$. Fortunately, if any $\hF_e^k$ such that spectrum of $L|_{\hF_e^k}$ falls into the strip $\hG(T_s)$, there is no system alias of the system.

	Based on Proposition \ref{infmain1}, obtaining $L|_{\hF_e^k}$ from $\bU|_{\hF_e^k}$ also requires a small enough period $T_s$ such that the spectrum $\sigma(L|_{\hF_e^k})$ falls into the strip $\hG(T_s)$, which is similar to the condition for point spectrum. If it holds, the principal logarithm of $U^{T_s}|_{\hF_e^k}$ refers to the generator $L|_{\hF_e^k}$. The critical sampling period of dynamical systems with Koopman continuous / residual spectrum are described as follows.
	
	\begin{theorem}[The critical sampling period $T_{\gamma,e}$]\label{infsampling}
		For the dynamical system \eqref{eq1} that has the (infinite-dimensional) Koopman invariant subspace in Assumption \ref{invariant}, if time-series $\{\bx(t_k) \}$ %$[\bx(t_k),\bx(t_{k+1})] $ 
		are equidistant sampled with sampling period $T_s$, to uniquely identify the vector field $f$ from the corresponding discrete flow $S^{T_s}$, the sampling frequency $w$ (rad/s) must satisfy \begin{equation*}
		w>2\min_{\hF_e^k}\{\max|\bIm(\sigma(L|_{\hF_e^k}))|\}.
		\end{equation*}
		Equivalently, the sampling period should satisfy $T_s<T_{\gamma,e}$, where the critical sampling period $T_{\gamma,e}$ (i.e., $2\pi/w$) can be computed by 
		\begin{equation*}
		\boxed{T_{\gamma,e} = \max_{\hF_e^k}\left\{\frac{\pi}{\max|\bIm(\sigma(L|_{\hF_e^k}))|}\right\}.}
		\end{equation*}
	\end{theorem}
	\vspace{1mm}
	\begin{proof}
		%{\proof	
		The result immediately follows by Proposition \ref{infmain1}.
		%}
	\end{proof}
	
	Theorem \ref{infsampling} provides the critical sampling period and the lower bound of sampling frequency with respect to system aliasing of dynamical systems with Koopman continuous / residual spectrum. Theorem \ref{th7} can be seen as its special case, where we can directly analyze the eigenfunction and define the finite-dimensional eigenspace. Although the form of the critical sampling periods in Theorem \ref{infsampling} is similar to the conclusion in Theorem \ref{th7}, it is difficult to analyze the systems of continuous or residual spectrum since the choice of $\hF_e^k$ and the spectrum could be complex. %Since the condition of the system aliasing needs to hold in only one Koopman invariant subspace $\hF_e^k$, $T_\gamma^e$ requires the ``best'' $\hF_e^k$, i.e., $\hF_e^\gamma$ defined below to allow the maximum of the sampling period bound: 	\begin{equation*}\hF_e^\gamma \triangleq \arg\max_{\hF_e^k}\left\{\frac{\pi}{\max|\bIm(\sigma(L|_{\hF_e^k}))|}\right\}.\end{equation*}
	\subsection{Critical sampling period for the motivating example}
	Let us take the extendable rod example in Section \ref{sec:pre} again. We will show in this example that: 1) the critical sampling period provides a sampling bound to avoid system aliasing; 2) the critical sampling periods which avoid aliasing of these dynamical systems are consistent with the Nyquist-Shannon sampling bounds to avoid aliasing of the associated states when the states are band-limited.
	
	Here we show in Fig. \ref{fig9} the key observation of system aliasing of the extendable rod, and use Theorem \ref{th7} to explian this phenomenon. Fig. \ref{fig9} illustrates the case that the length of the rod is getting longer ($\dot{r} = 0.1r$) with the angular velocity $\omega=3$ rad/s. The real state of time $x_1(t)$ is denoted by the blue line, and the false state of time $\hat{x}_1(t)$ is denoted by the orange line, which is simulated by the identified dynamical systems with sampling period $T_s = 4\pi/9$ s. We find $\hat{x}_1(t)$ matches the same measurements of the correct system (marked by black triangles). However, it shows that $x_1(t)\neq\hat{x}_1(t)$ when $t\neq kT_s, k=1,2,\ldots$. According to Theorem \ref{th7}, the critical sampling period of the dynamical system is $T_\gamma = \pi/3$ s, which is denoted by the yellow dashed line. If the sampling period $T_s>T_\gamma$, there are generator alias $\widehat{L}$, which corresponds to system aliases $\hat{f}$, in aliasing space besides the truth. It causes that different CT generator can match the same DT Koopman operator $U^{T_s}$, i.e., different vector fields can generate the same measurements as in Fig. \ref{fig9}. Therefore, the critical sampling period $T_\gamma$ provides the maximum sampling period to avoid the aliasing of the dynamical system caused by sampling period, %guarantee that the true vector field of the dynamical system $f$ is the unique one in the aliasing space,
	i.e., the sampling period needs to be small enough ($T_s<T_\gamma$) such that the uniquely identified dynamical system corresponds to the truth.

%	\begin{figure}[!t]
%			\centering
%			\subfigure[The length of the rod is constant $\dot{r} = 0$]{
%			\label{Fig9.sub1}
%			\includegraphics[width=.5\textwidth]{alias_explain_0.pdf}}
%			\subfigure[The length of the rod is getting longer $\dot{r}=0.1r$]{
%			\label{Fig9.sub2}
%			\includegraphics[width=.5\textwidth]{alias_explain_1.pdf}}
%			\caption{The length of the rod is getting longer $\dot{r}=0.1r$. The true state $x_1(t)$ of the dynamical system \eqref{eq1} which corresponds to the true generator $L$, and the state $\hat{x}_(t)$ predicted by the generator alias $\widehat{L}$, where $U^{T_s} =\exp(T_sL), \widehat{L} = \Log(U^{T_s})/T_s$. The critical sampling period of the system is $T_\gamma$. }
%			\label{fig9}
%	\end{figure}
	\begin{figure}[!t]
	\centering
		\includegraphics[width=.5\textwidth]{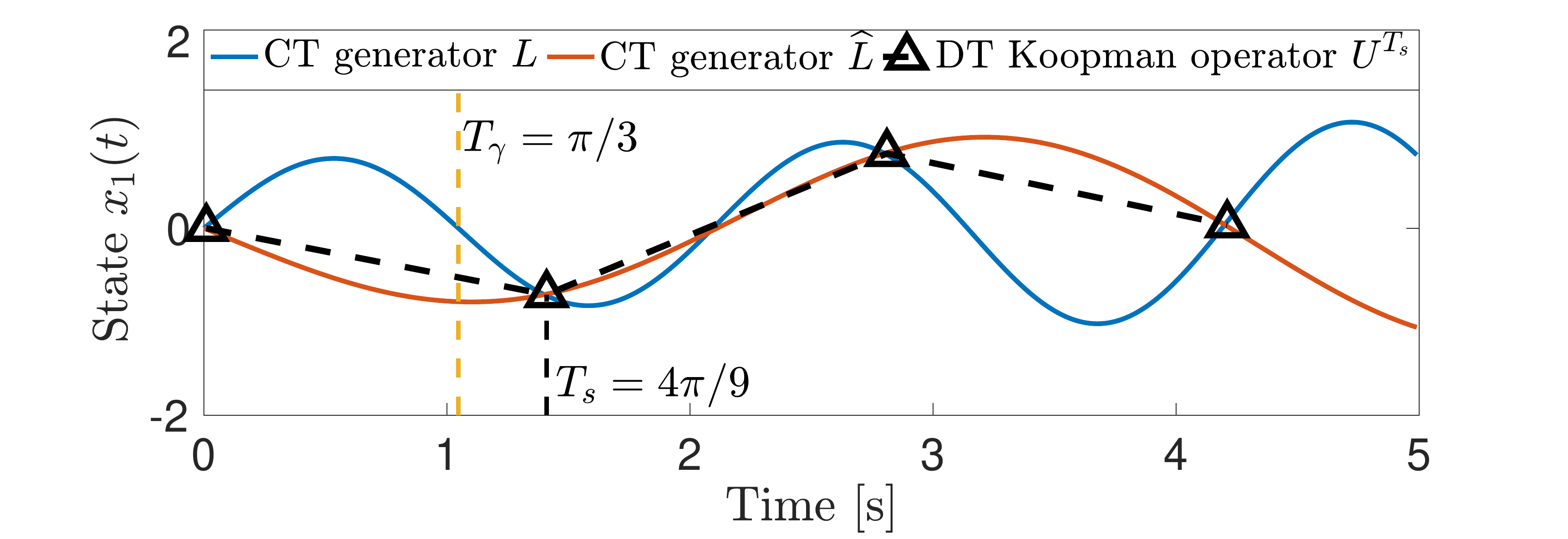}
	\caption{The length of the rod is getting longer $\dot{r}=0.1r$. The true state $x_1(t)$ of the dynamical system \eqref{eq1} which corresponds to the true generator $L$, and the state $\hat{x}_(t)$ predicted by the generator alias $\widehat{L}$, where $U^{T_s} =\exp(T_sL), \widehat{L} = \Log(U^{T_s})/T_s$. The critical sampling period of the system is $T_\gamma$. }
	\label{fig9}
\end{figure}
	
	Since the Nyquist-Shannon sampling theorem also considers the sampling bound of signals, we use these two cases of extendable rod to analyze the Nyquist-Shannon sampling bound of the states, and the critical sampling period of the associated systems. When the length of the rod is constant and the angular velocity is $\omega$ rad/s, the trajectory of $x_1(t)$ and $x_2(t)$ are periodic with the period $2\pi/\omega$ s. The lower bound of sampling frequency by Nyquist-Shannon sampling theorem is $2\omega$ rad/s. Accordingly, the dynamical system is a linear system described by \begin{align*}
	\dot{x}_1 &= \omega x_2,\\%\quad
	\dot{x}_2 &= -\omega x_1.
	\end{align*} Note that the Koopman operator of this dynamical system has eigenvalues $\pm \omega i$. Therefore, the lower bound of frequency given by Theorem \ref{th7} is $$
	2\min_{\hF_2}\{\max|\bIm(\lambda(L|_{\hF_2}))|\}=2\omega,
	$$ which is consistent with the bound provided by Nyquist-Shannon sampling theorem.
	
	When the length of rod is getting longer ($\dot{r} = 0.1r$) while rotating with the constant angular velocity $\omega$ rad/s, as in Fig. \ref{fig1}, the dynamical system can also be described by \begin{align*}
	\dot{x}_1 &= 0.1x_1+\omega x_2,\\%\quad
	\dot{x}_2 &= -\omega x_1+ 0.1x_2.
	\end{align*}
	The lower bounds of sampling frequency given by Theorem  \ref{th7} is still $2\omega$ rad/s. However, Nyquist-Shannon sampling theorem is not applicable in this case since the states of $x_1(t)$ and $x_2(t)$ are not band-limited signals.	
	%\ye{add}
	 %\ye{integrable? You never mention this any where}
	
	%----------------------------------------------------------
	
	%----------------------------------------------------------------------------	

	\section{Numerical experiments}\label{sec:num}
	%\ye{put a link to the github pdf}
	In this section, the identification errors with increasing sampling periods are presented to verify the sampling criterion (Theorem  \ref{th7}). Koopman-based approach \cite{mauroy2019koopman} is used for identification since it does not need to calculate the derivatives of the states and can identify dynamical systems using measurements with lower sampling frequency. To test the influence of sampling period on the identification, the errors of identified vector fields $\hat{f}$ with growing sampling periods are analyzed. 
	
	We first describe the Koopman-based identification approach and then show the numerical example of a nonlinear system with limit cycle. Here is the identification approach. %\ye{please change all upper bound throughout the paper, there is only critical sampling period $T_\gamma$}
	
	Consider the state vector $\bx = [x_1,\ldots,x_n]$ and the following dynamical system:
	\begin{equation*}
	\begin{aligned}
	\dot{x}_k = f_k(\bx) &= \sum_{j = 1}^{N}w_j^k g_j(\bx), \ k = 1,\ldots,n,\\
	\end{aligned}
	\end{equation*}
	where $g_j(\bx)$ are monomial basis functions and coefficients $w_j^k$ are to be identified. 
	
	Numerically, we define an $N$-dimensional observable space $\hF_m = \sspan\{g_1, g_2, \ldots, g_N \}$ to approximate the Koopman invariant subspace, where $m$ denotes the total degree of the monomial basis function. The observable space is defined as $$g_j(\bx) \in \{x^{s_1}_1\ldots x_n^{s_n}|(s_1,\ldots ,s_n) \in \bN^n, s_1+\ldots +s_n\le m\},$$ where $x_j (1\le j\le n)$ is the $j$th component of $\bx$. %\ye{n has been used before} 
	The data $\{\bx(t_k),\bx(t_{k+1}) \}_{k=1}^K$, where $t_{k+1} - t_{k} = T_s$, are lifted to this approximate Koopman invariant subspace $\hF_m$, i.e.,
	%$[X_{\text{lift}}]_{ij} = g_j(\bx(t_i)),[Y_{\text{lift}}]_{ij} = g_j(\bx(t_{i+1}))$.
	\begin{equation*}
X_{\text{lift}} = \left(\begin{matrix}
g_1(\bx(t_1))&\ldots&g_N(\bx(t_1))\\
\vdots&\ddots&\vdots\\
g_1(\bx(t_K))&\ldots&g_N(\bx(t_K))\\
\end{matrix}
\right)_{K\times N},
	\end{equation*}
	\begin{equation*}
	Y_{\text{lift}} = \left(\begin{matrix}
	g_1(\bx(t_2))&\ldots&g_N(\bx(t_2))\\
	\vdots&\ddots&\vdots\\
	g_1(\bx(t_{K+1}))&\ldots&g_N(\bx(t_{K+1}))\\
	\end{matrix}
	\right)_{K\times N}.
	\end{equation*}
	
	In order to focus on the influence of the sampling period on the identification, we collect enough data and $K\gg N$. Since $\hF_m$ is an approximate Koopman invariant space, the lifting data $X_{\text{lift}}$ and $Y_{\text{lift}}$ are approximately linear, the matrix representation $\widehat{U}$ of the Koopman operator can be approximately obtained by \begin{equation*}
	\widehat{U} = X_{\text{lift}}^\dagger Y_{\text{lift}}.
	\end{equation*}
	Then the approximate matrix representation of the generator is computed as \begin{equation*}
	\widehat{L} = \frac{1}{T_s}\Log(X_{\text{lift}}^\dagger Y_{\text{lift}}).
	\end{equation*}
	To recover the vector field, i.e., identify the coefficients $w_j^k$, we use the basis function $g_l(\bx)=x_k$ and $Lg_l = (f\cdot\nabla )g_l = f_k$. It follows that the coefficients of the vector field can be recovered from the matrix representation of the generator $L$, i.e., $\hat{w}^k_j=[\widehat{L}]_{jl}$, where $l$ is the index of the basis function $g_l(x)=x_k$. 
	
	The normalized root-mean-square error (NRMSE) of coefficients is computed to describe the identification error of the vector field, i.e., \begin{equation*}
	\mathrm{RMSE} = \sqrt{\frac{1}{n N}\sum_{k =1}^{n}\sum_{j=1}^{N}(\hat{w}_j^k-w_j^k)^2},
	\end{equation*} \begin{equation*}
	\mathrm{NRMSE} =\mathrm{RMSE}/|w|,
	\end{equation*}where $|w|$ is the average value of all nonzero coefficients $w_j^k$ of associated dynamical systems, i.e., $
	|w| = \sum_{k=1}^{n}\sum_{j=1}^{N}|w_j^k|/\|w\|_0.
	$ %\ye{add a mathematical formula for $|w|$}

	The dynamical systems for identification are as follows:
	
	\emph{a) The nonlinear system with real Koopman eigenvalues.}
	\begin{equation}\label{sys4}
	\begin{aligned}
	\dot{x}_1 &= - x_1,\\
	\dot{x}_2 &= x_1^2 - x_2.
	\end{aligned} 
	\end{equation}
	
	In the Koopman-based identification approach, we set 1000 trajectories, which is the same set up with \cite{korda2018linear}, 30 snapshots at times for each trajectory and the random initial condition $\bx(t_0)\in [-1,1]^2$. In order to minimize the impact of other factors such as the choice of observable space on error of identification, we select the smallest Koopman invariant space, i.e., $ \hF_m = \sspan\{g_1(\bx) = x_1, g_2(\bx) = x_2, g_3(\bx) = x_1^2\}$. 
	
	The lower bound of sampling frequency based on Theorem \ref{th7} is analyzed as follows. For the Koopman invariant space $\hF_m$, the generator associated with the system \eqref{sys4} can be seen as a three-dimensional linear system:
	$$
	\frac{\rm d}{{\rm d} t}\left(\begin{matrix}
	x_1\\x_2\\x_1^2
	\end{matrix}
	\right)
	= \left(\begin{matrix}
	-1&0&0\\0&-1&1\\0&0&-2
	\end{matrix}
	\right) \left(\begin{matrix}
	x_1\\x_2\\x_1^2
	\end{matrix}\right)
	$$
	Thus the eigenvalues are $-1$ and $-2$, which are all real values. According to Theorem \ref{th7}, the theoretical lower bound of sampling frequency is zero and the critical sampling period is not bounded.
	
	Fig. \ref{Fig7.sub1} shows the NRMSE of the dynamical system. The blue line denotes the result of NRMSE with growing sampling periods. It shows that NRMSE keeps small with growing sampling period, which is in agreement with Theorem \ref{th7}. The states $\hat{\bx}(t)$ simulated by the identified dynamical system and the real states $\bx(t)$ are visualized in Fig. \ref{Fig7.sub2}-Fig. \ref{Fig7.sub4} when the sampling periods are $0.5$s, $1.1$s and $2.8$s. The blue lines denote the true states of time $\bx(t)$ with the initial condition $[0.5,0.5]$. The orange dashed lines are the prediction of states $\hat{\bx}(t)$ by identified vector fields $\hat{f}$. It implies that the true states of time $\bx(t)$ is always the ''simplest'' trajectory between two measurements, which may be the reason why this system does not have the bounded critical sampling period.
	
	\begin{figure}[thpb]
		\subfigure[Identification error]{
			\label{Fig7.sub1}
			\includegraphics[width=.225\textwidth]{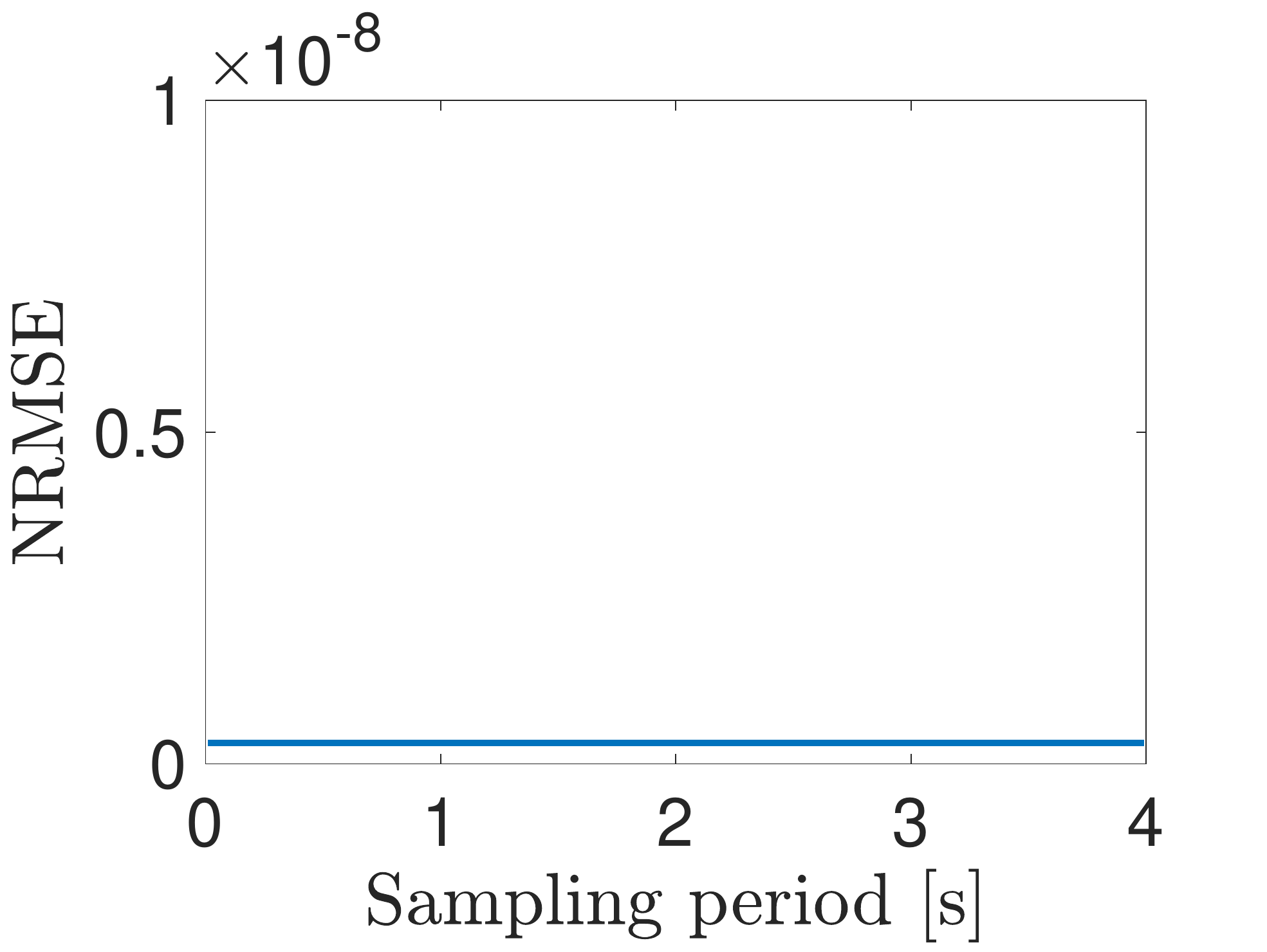}}
		\subfigure[State prediction ($T = 0.5$ s) ]{
			\label{Fig7.sub2}
			\includegraphics[width=.225\textwidth]{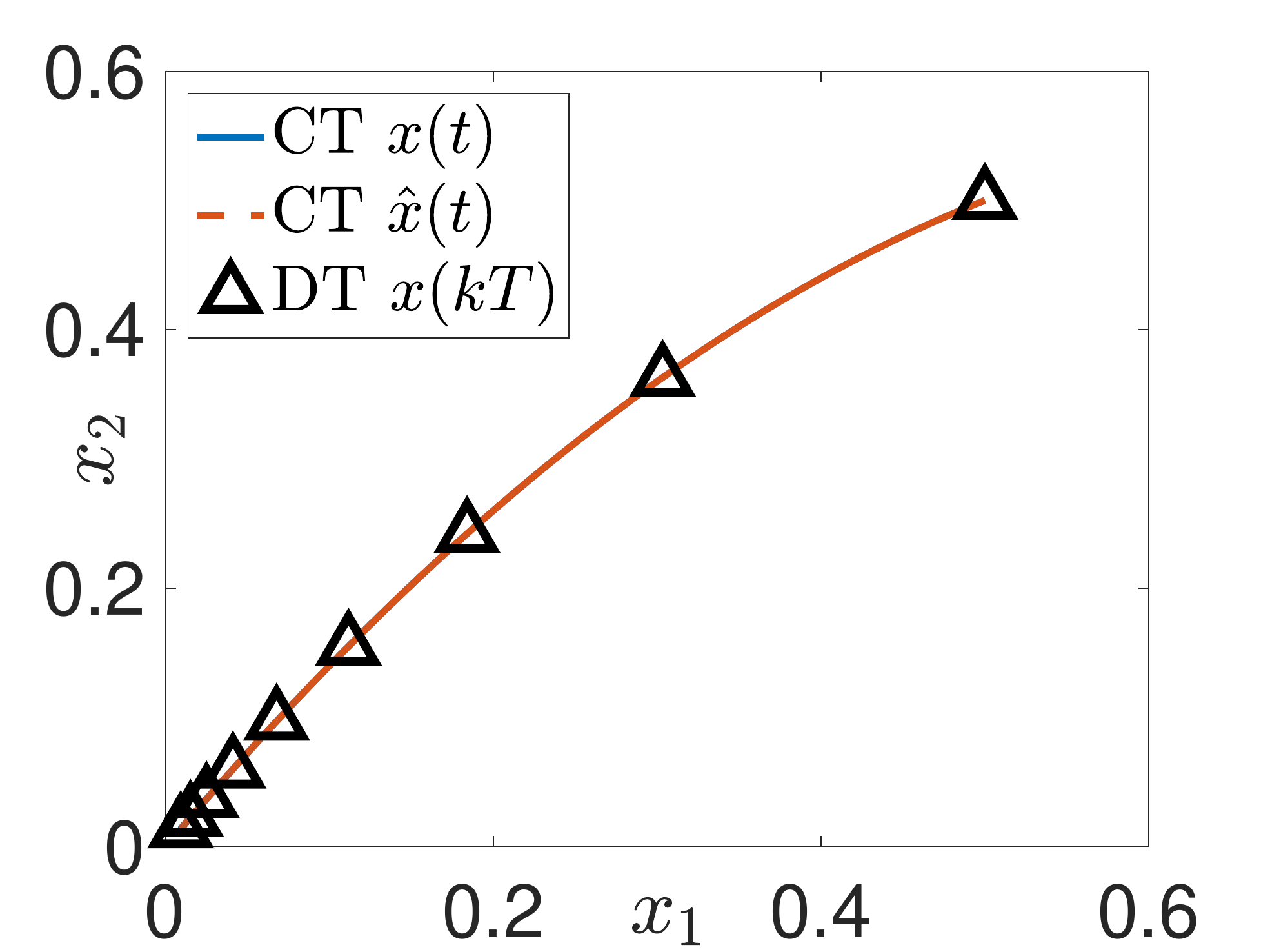}}
		\subfigure[State prediction ($T = 1.1$ s)]{
			\label{Fig7.sub3}
			\includegraphics[width=.225\textwidth]{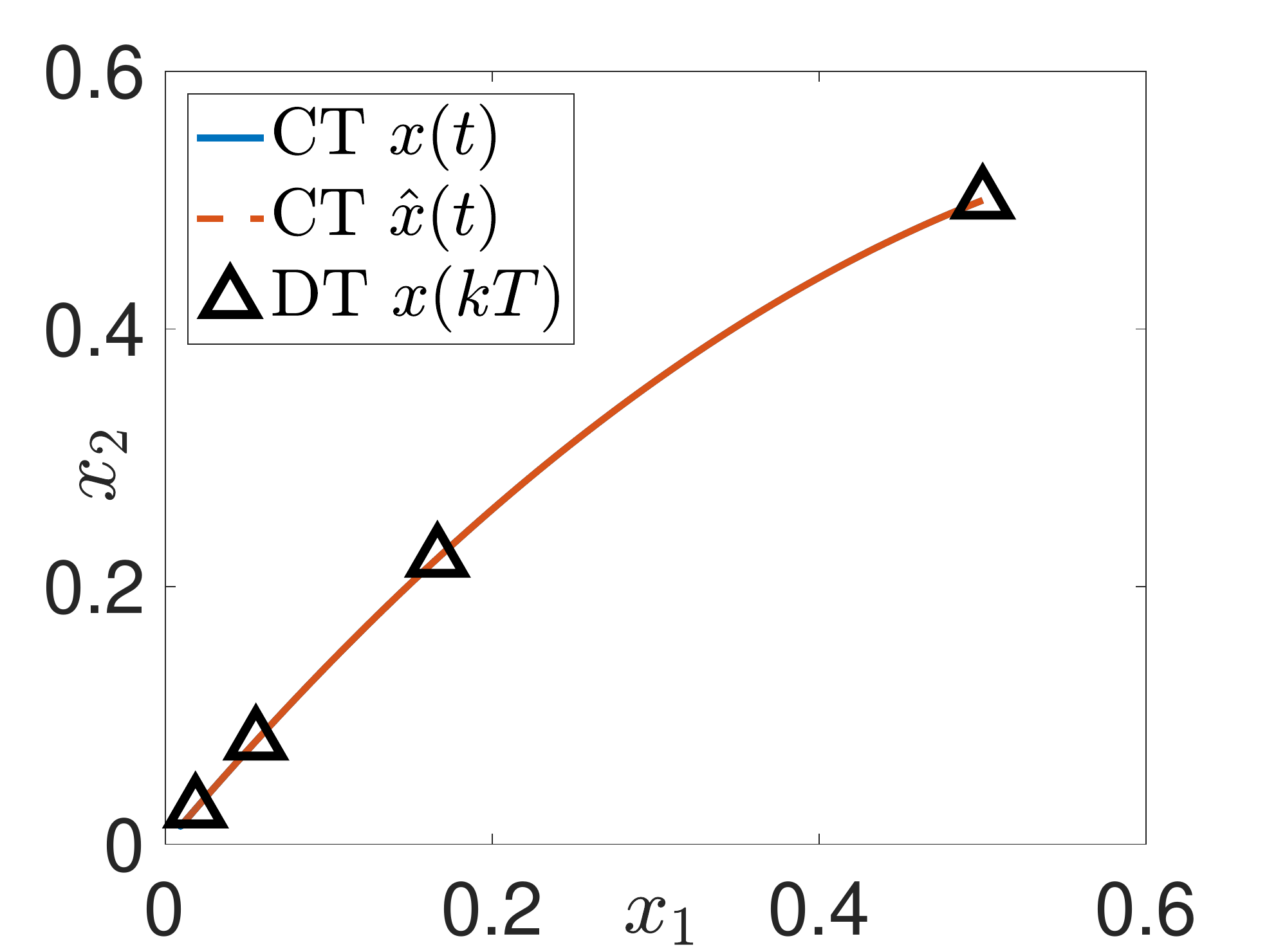}
		}
		\subfigure[State prediction ($T = 2.8$ s)]{
			\label{Fig7.sub4}
			\includegraphics[width=.225\textwidth]{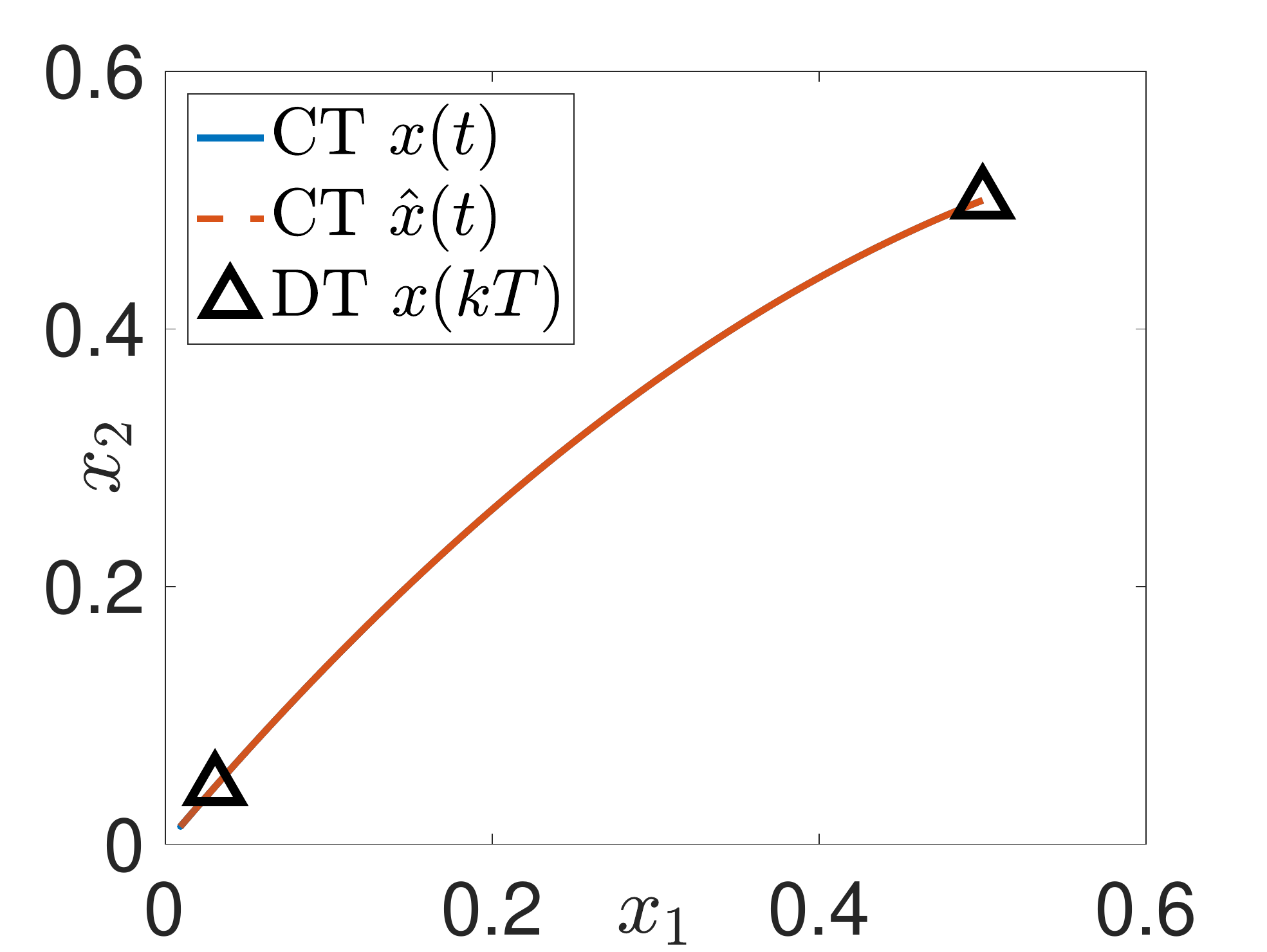}
		}
		\caption{Identification errors with growing sampling periods (a), and prediction by vector field identified from data with sampling period $T=0.5$ s (b), $T=1.1$ s (c), and $T=2.8$ s (d) of nonlinear system \eqref{sys4} whose Koopman eigenvalue are real.}
		\label{fig7}
	\end{figure}
	
	\emph{b) The linear system with states $\bx = [x_1,x_2]^T$.}
	\begin{equation}\label{sys1}
	\begin{aligned}
	\dot{x}_1 &= 0.1x_1+3x_2,\\
	\dot{x}_2 &= -3x_1+ 0.1x_2.
	\end{aligned} 
	\end{equation}
	
	In the Koopman-based identification approach, we set 1000 trajectories, 30 snapshots at times for each trajectory, and the random initial condition $\bx(t_0)\in [-1,1]^2$. %In order to minimize the impact of other factors such as the choice of observable space on error of identification, 
	We select the smallest Koopman invariant space for identification, i.e., $m=1$. 
	
	The critical sampling period based on Theorem \ref{th7} is analyzed as follows. For the linear system \eqref{sys1}, the eigenspace $\hF_2$ is the space of state observable functions $\sspan\{g_1(\bx) = x_1,g(\bx)=x_2\}$. In this space, the generator associated with the system \eqref{sys1} can be represented by the following system matrix with the basis functions $g_1, g_2$.
	$$
	L|_{\hF_2} = \left(\begin{matrix}
	0.1&3\\-3&0.1
	\end{matrix}
	\right)
	$$
	Thus the eigenvalues are determined by \begin{equation*}
	(\lambda-0.1)^2+3^2=0,
	\end{equation*}leading to $\lambda_{1,2}=0.1\pm3i.$ According to Theorem \ref{th7}, %the theoretical lower bound of sampling frequency is $2\max|\bIm(\lambda(L|_{\hF_2}))|=6$ rad/s and 
	the critical sampling period is $\pi/\max|\bIm(\lambda(L|_{\hF_2}))|=\pi/3$ s. 
	The lower bound of sampling frequency will be integer multiples of $6$ rad/s if other Koopman valid eigenspace $\hF_2$, spanned by multiplications of these eigenfunctions, are used to analyze the lower bound based on Property \ref{property1}. 
	
	Fig. \ref{Fig4.sub1} shows the comparison of critical sampling period and NRMSE of the linear dynamical system. The blue line shows the result of NRMSE with growing sampling periods. The red dashed line denotes the critical sampling period $\pi/3$ s, which is computed by Theorem \ref{th7}. It is in agreement with Theorem \ref{th7} that the NRMSE changes significantly as $T_s$ tends to the theoretical red line.

	\begin{figure}[thpb]
		\subfigure[Identification error]{
			\label{Fig4.sub1}
			\includegraphics[width=.225\textwidth]{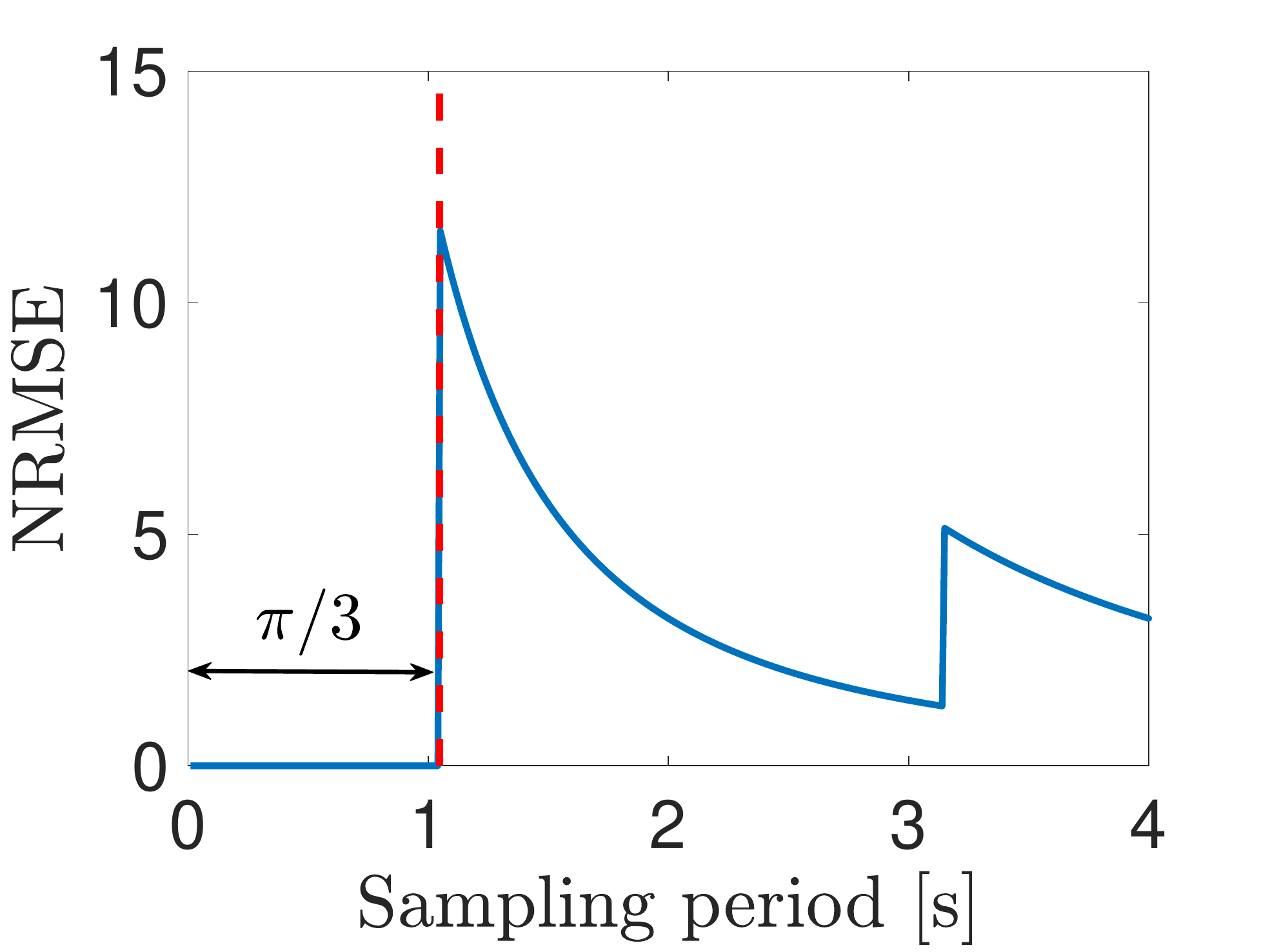}}
		\subfigure[State prediction ($T = 0.5$ s) ]{
			\label{Fig4.sub2}
			\includegraphics[width=.225\textwidth]{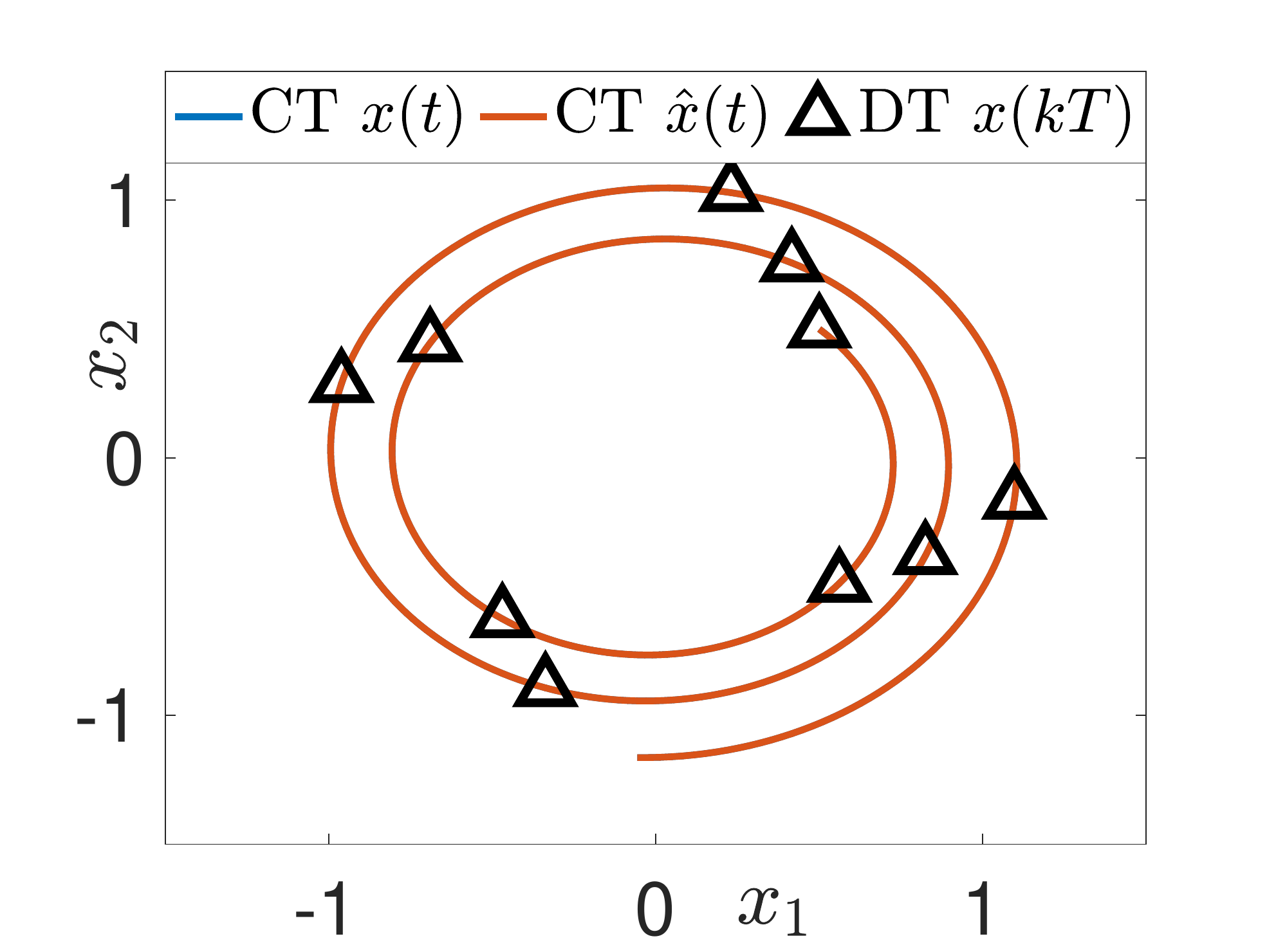}}
		\subfigure[State prediction ($T = 1.1$ s)]{
			\label{Fig4.sub3}
			\includegraphics[width=.225\textwidth]{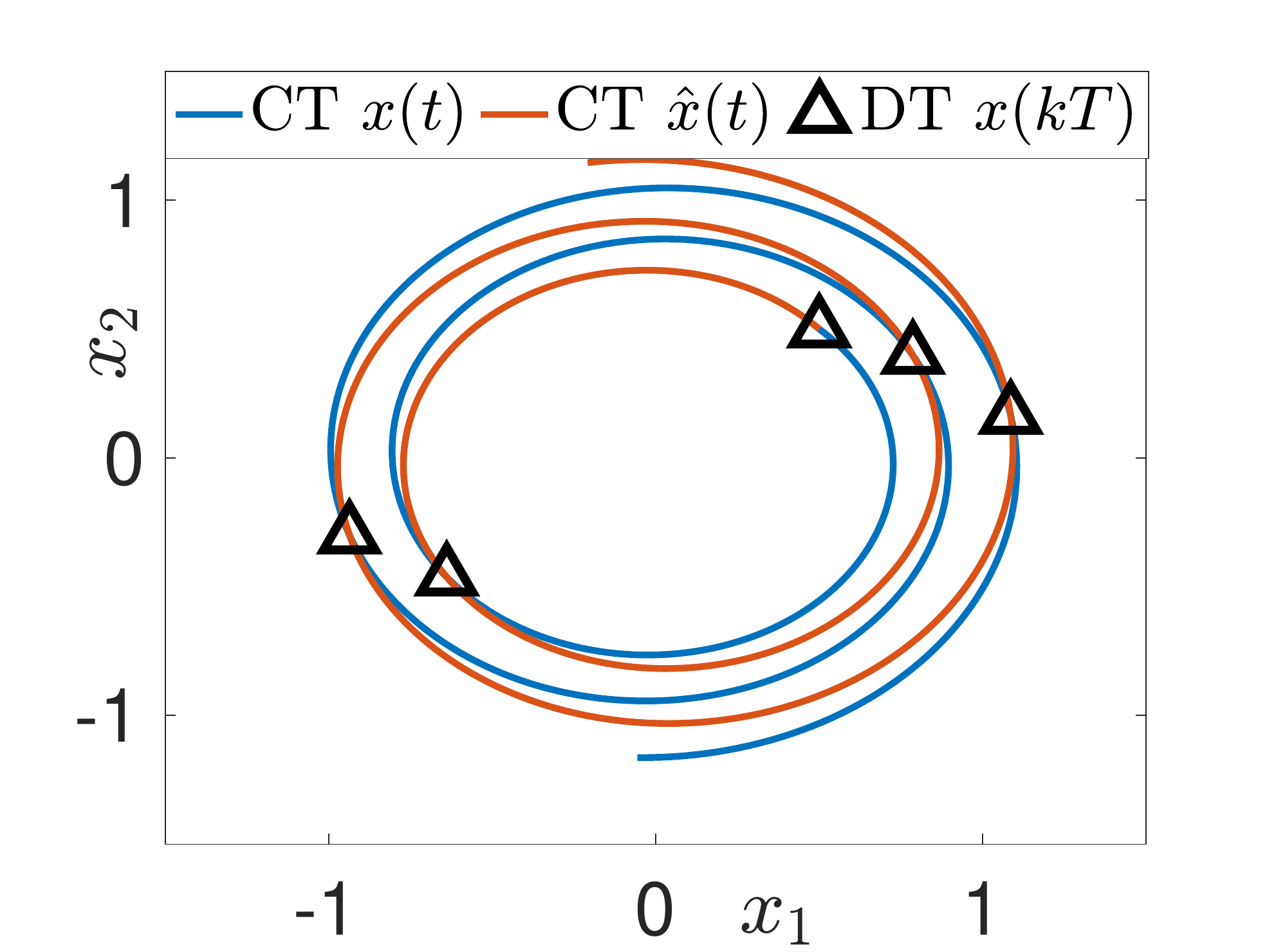}
		}
		\subfigure[State prediction ($T = 2.8$ s)]{
			\label{Fig4.sub4}
			\includegraphics[width=.225\textwidth]{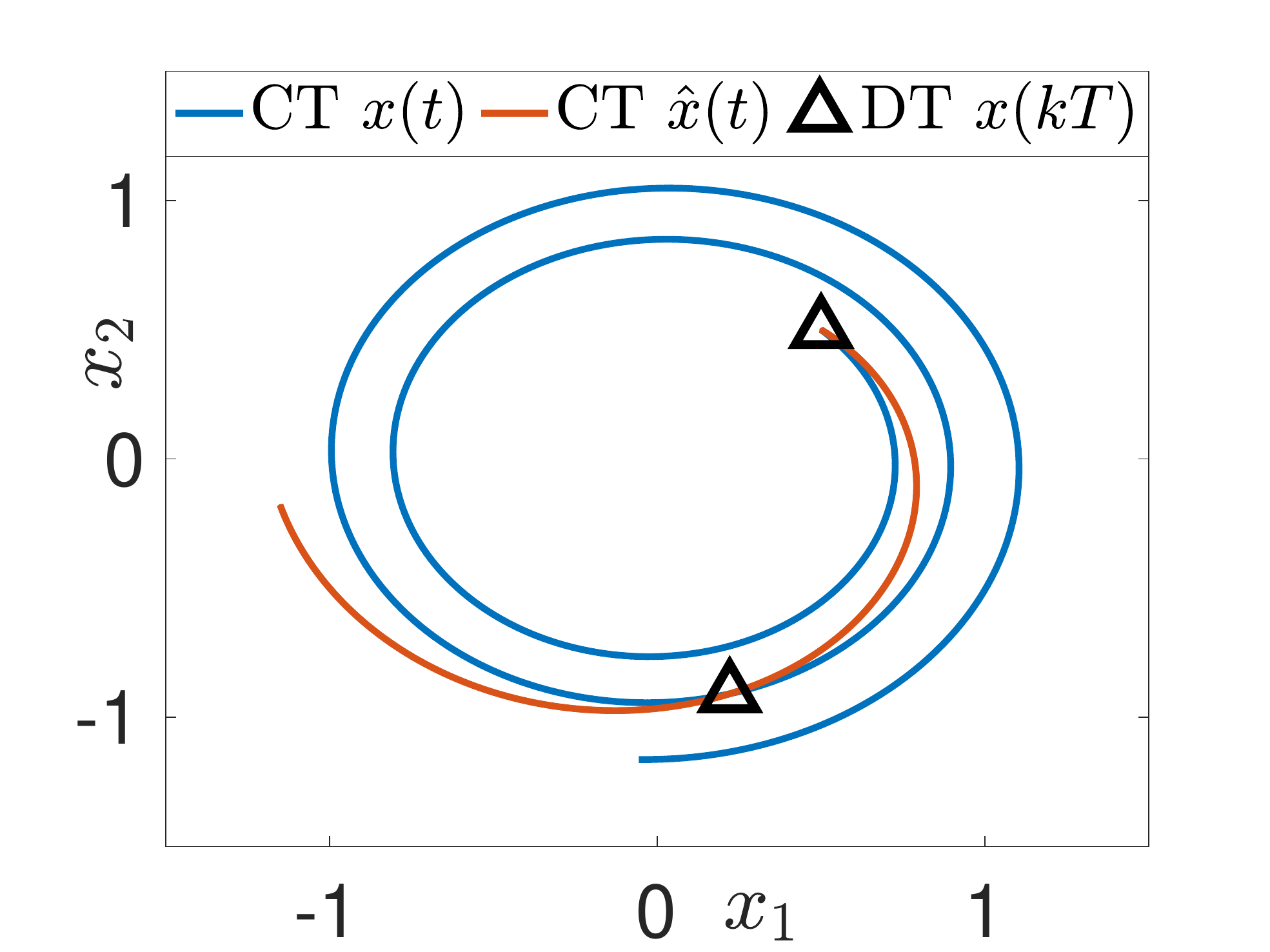}
		}
		\caption{Identification errors with growing sampling periods (a), and prediction by vector field identified from data with sampling period $T=0.5$ s (b), $T=1.1$ s (c), and $T=2.8$ s (d) of the linear system \eqref{sys1}. The parameter of the observable space is $m=1$ for prediction.}
		\label{fig4}
	\end{figure}
	
	\emph{c) The nonlinear system with a fixed point at $[0,0]$}:
	\begin{equation}\label{sys2}
	\begin{aligned}
	\dot{x}_1 &= -3x_2-x_1(x_1^2+x_2^2),\\
	\dot{x}_2 &= 3x_1-x_2(x_1^2+x_2^2).
	\end{aligned}
	\end{equation}
	
	We set 1000 trajectories, 30 snapshots at times for each trajectory and random initial condition $\bx(t_0)\in [-1,1]^2$ for identification. In order to make the comparison of sampling period bound and the change of NRMSE with growing sampling period more convincing, three eigenspaces $\hF_m, m = 7,10,13$ are chosen to approximate the Koopman invariant space for the nonlinear system \eqref{sys2}. 
	
	For nonlinear systems with fixed point $x^*$, i.e., $f(x^*) = 0$, the principal Koopman eigenvalues $\lambda_i$ are the eigenvalues of the Jacobian matrix of the vector field $f$ at the fixed point $x^*$ \cite{mauroy2016global}. Therefore, eigenvalues of the generator of the system \eqref{sys2} are eigenvalues of the Jacobian matrix at $[0,0]$: $$
	J = \left(\begin{matrix}
	0&-3\\3&0
	\end{matrix}\right).
	$$
	It leads to the principal Koopman eigenvalues being $\pm3i$. According to Theorem \ref{th7}, the lower bound of sampling frequency is $6$ rad/s and the critical sampling period is $\pi/3$ s. Other Koopman eigenvalues can be obtained by taking linear combinations (with weights that are positive and integers) of these principal eigenvalues. The bound of sampling frequency analyzed by associated eigenfunctions will be integer multiples of $6$ rad/s.
	
	The critical sampling period $\pi/3$ is denoted as the red dashed line in Fig. \ref{Fig5.sub1}. It shows that the NRMSE$^{1/4}$ of the nonlinear system \eqref{sys2} appears to a peak when the sampling period grows close to the red dashed line in every observable space we choose, which is in agreement with Theorem \ref{th7}. 
	
	\begin{figure}[thpb]
		\subfigure[Identification error]{
			\label{Fig5.sub1}
			\includegraphics[width=.225\textwidth]{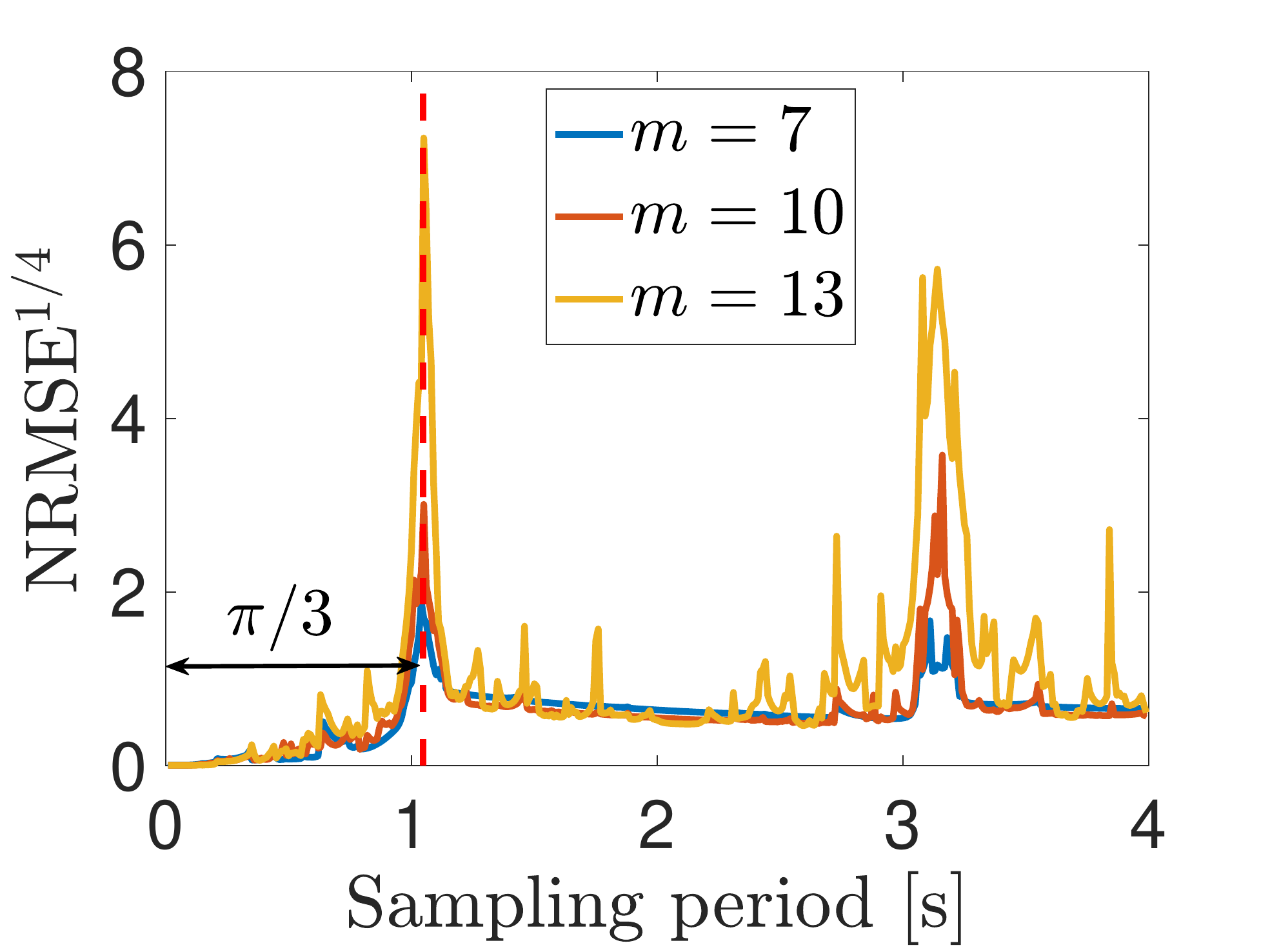}}
		\subfigure[State prediction ($T = 0.5$ s)]{
			\label{Fig5.sub2}
			\includegraphics[width=.225\textwidth]{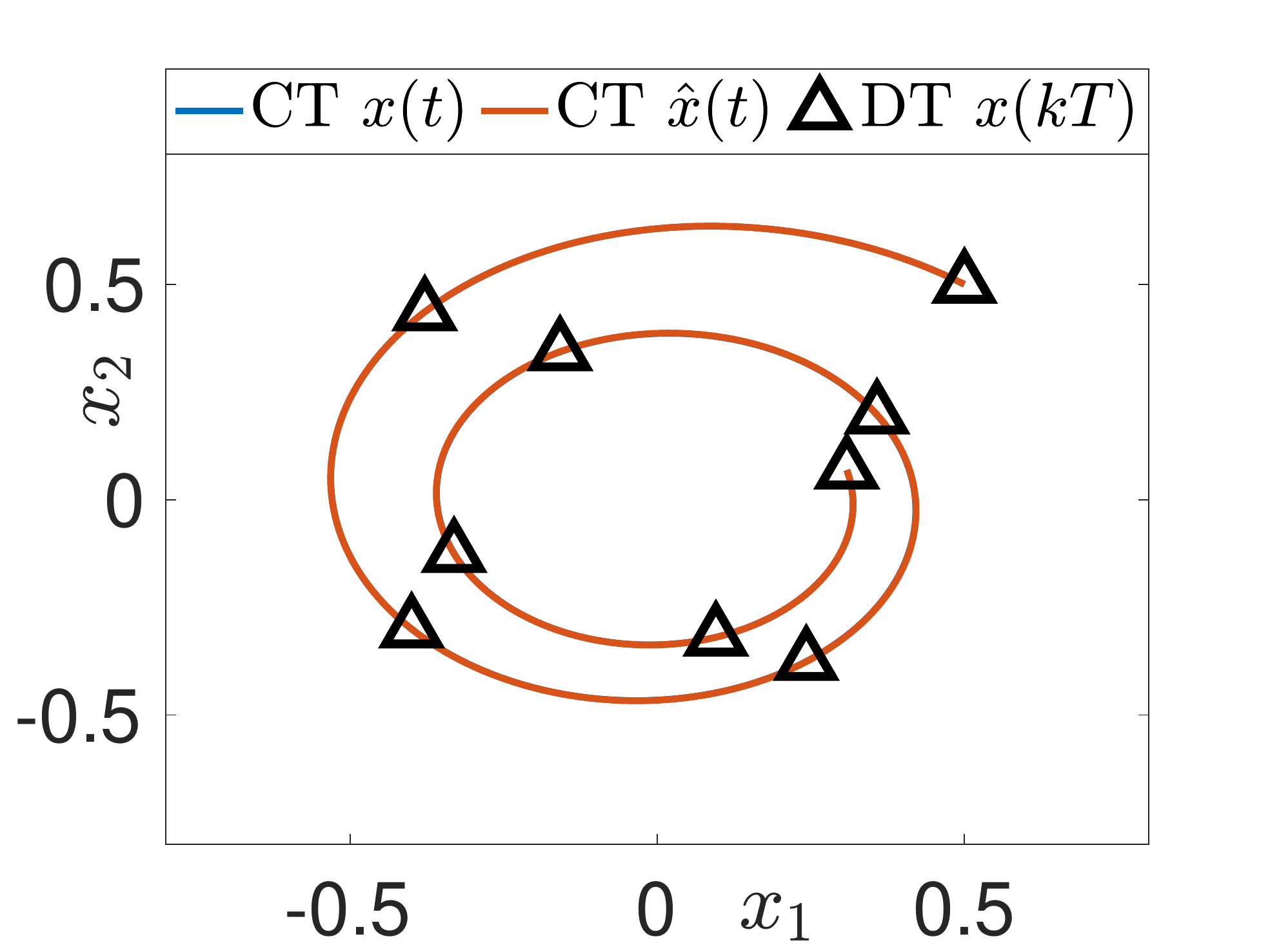}}
		\subfigure[State prediction ($T = 1.1$ s)]{
			\label{Fig5.sub3}
			\includegraphics[width=.225\textwidth]{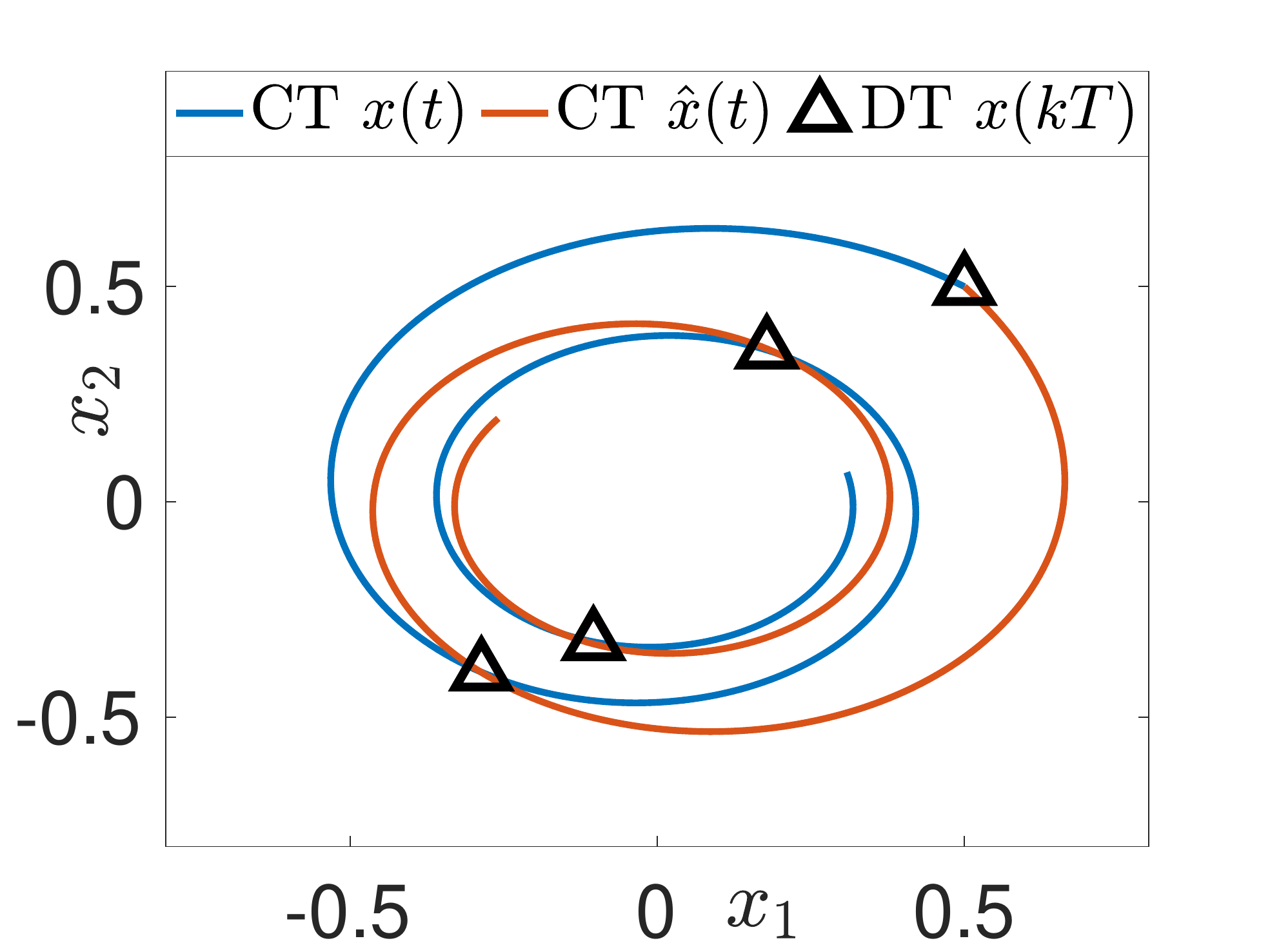}
		}
		\subfigure[State prediction ($T = 2.8$ s)]{
			\label{Fig5.sub4}
			\includegraphics[width=.225\textwidth]{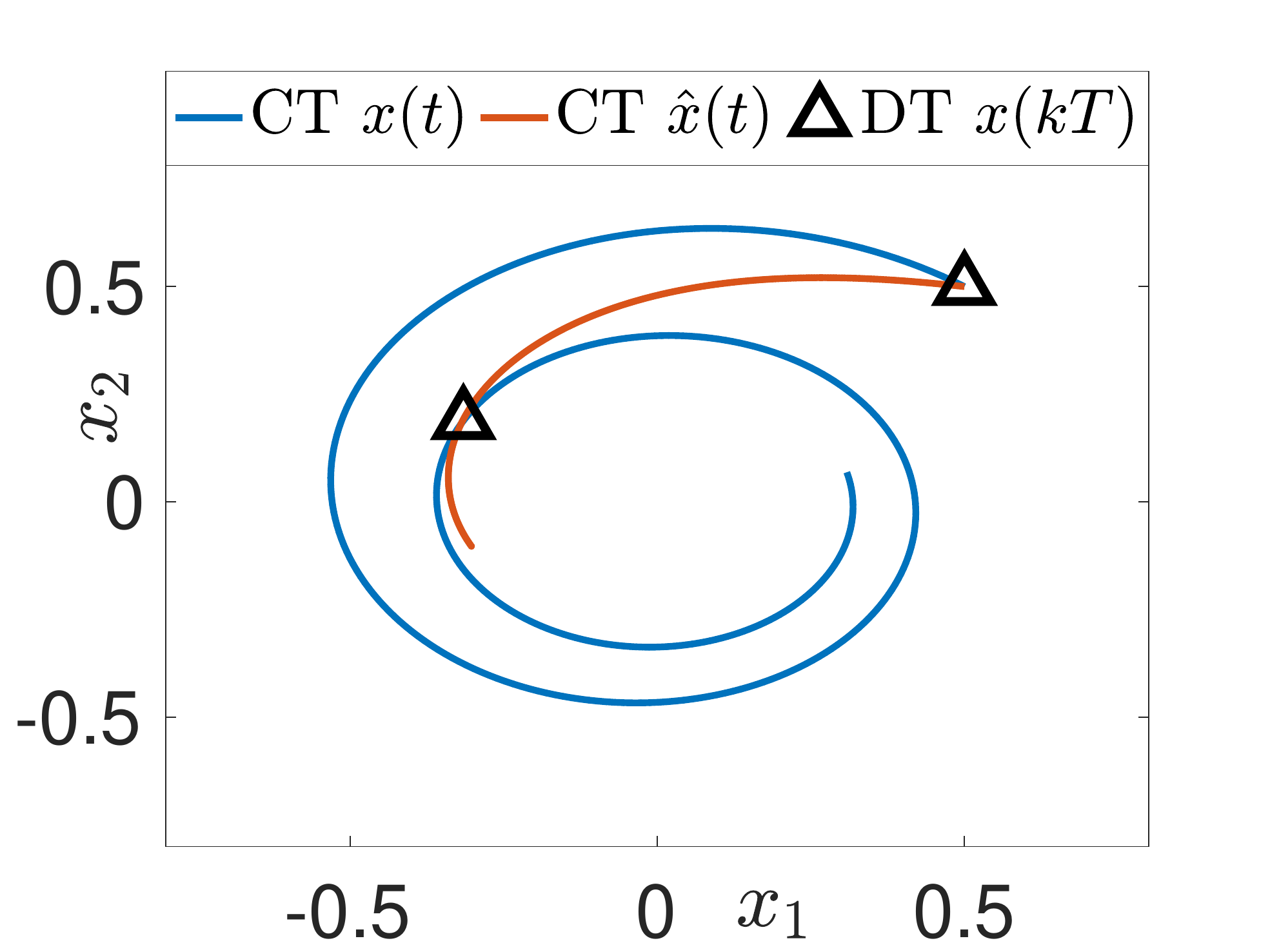}
		}
		\caption{Identification errors with growing sampling periods (a), and prediction by vector field identified from data with sampling period $T=0.5$ s (b), $T=1.1$ s (c), and $T=2.8$ s (d) of nonlinear system with fixed point \eqref{sys2}. The parameter of the observable space is $m=13$ for prediction.}
		\label{fig5}
	\end{figure}

	\emph{d) The nonlinear system with a limit cycle at $r = \sqrt{x_1^2+x_2^2}=1$:}
	\begin{equation}\label{sys3}
	\begin{aligned}
	\dot{x}_1 &= 3x_2-x_1(x_1^2+x_2^2-1),\\
	\dot{x}_2 &= -3x_1-x_2(x_1^2+x_2^2-1).
	\end{aligned} 
	\end{equation}
	
	We set 1000 trajectories, 30 snapshots at times for each trajectory and random initial condition $\bx(t_0)\in [-1.5,1.5]^2$ to identify this system. Moreover, three eigenspaces $\hF_m, m = 7,10,13$ are used to approximate the Koopman invariant space for the nonlinear system \eqref{sys3}.
	
	The principal Koopman eigenvalues are calculated as follows. Using polar coordinates, i.e., $x_1 = r\cos \theta$ and $x_2 = r\sin \theta$, we obtain the dynamical system
	\begin{align*}
	\dot{r} &= -r^3+r,\label{eq10}\\
	\dot{\theta} &= -3.
	\end{align*}
	For dynamical systems that admit limit cycles, the Floquet exponents are principal Koopman eigenvalues \cite{mauroy2016global}. Since the limit cycle is at $r = 1$, %it becomes $\dot{\hat{r}} = -\hat{r}^3-3\hat{r}^2-2\hat{r}$ with $\hat{r}=r-1$ from \eqref{eq10}. Then 
	the eigenvalue (Floquet exponent) is $-2$. %Since trajectory of $x_1$ and $x_2$ are oscillatory and $\theta$ that be observed is correspondingly recurrent, the observable function of $\theta$ should capture this property. 
	Other principal eigenvalues are $\pm3i$ with associated eigenfunction $e^{\mp i\theta}$. According to Theorem  \ref{th7}, the critical sampling period is $\pi/3$ s. 
	
	Fig. \ref{Fig6.sub1} shows the NRMSE$^{1/4}$ of the identified vector fields in three spaces ($m=7,10,13$) and the critical sampling period $\pi/3$ s (denoted by the red dashed line). We notice that the NRMSE$^{1/4}$ change significantly as sampling periods tend to ($m=7,10$) or a little earlier ($m=13$) than $\pi/3$ s, which are consistent with Theorem \ref{th7}. Furthermore, this numerical example ($m=13$) also suggests that it may not always be better to use more basis functions in Koopman-based identification method. %On the one hand, a larger space does not mean that it approximates the Koopman invariant subspace better. If it does not approximate Koopman invariant subspace well, then the identification is more likely to become worse as the sampling period increases.
	 When the number of basis functions increases as $m$ gets larger, this method may have higher requirement for sampling period. Since the eigenfunction that is the product of principal eigenfunctions may belong to $\hF_m$ when $m$ gets larger. The associated eigenvalue, which is linear combinations of principal eigenvalues, results in $\max{|\bIm(\sigma_p(L|_{\hF_m}))|}$ getting larger, where $L|_{\hF_m}$ is the approximate matrix representation of the generator restricted to $\hF_m$. Therefore, the sampling period $T_s$ that allows valid identification, i.e., $L|_{\hF_m}T_s$ being the principal logarithm of the Koopman operator, may be smaller with larger $m$, and the NRMSE may jump earlier than the line of critical sampling period $T_\gamma$.
	\begin{figure}[thpb]
		\subfigure[Identification error]{
			\label{Fig6.sub1}
			\includegraphics[width=.225\textwidth]{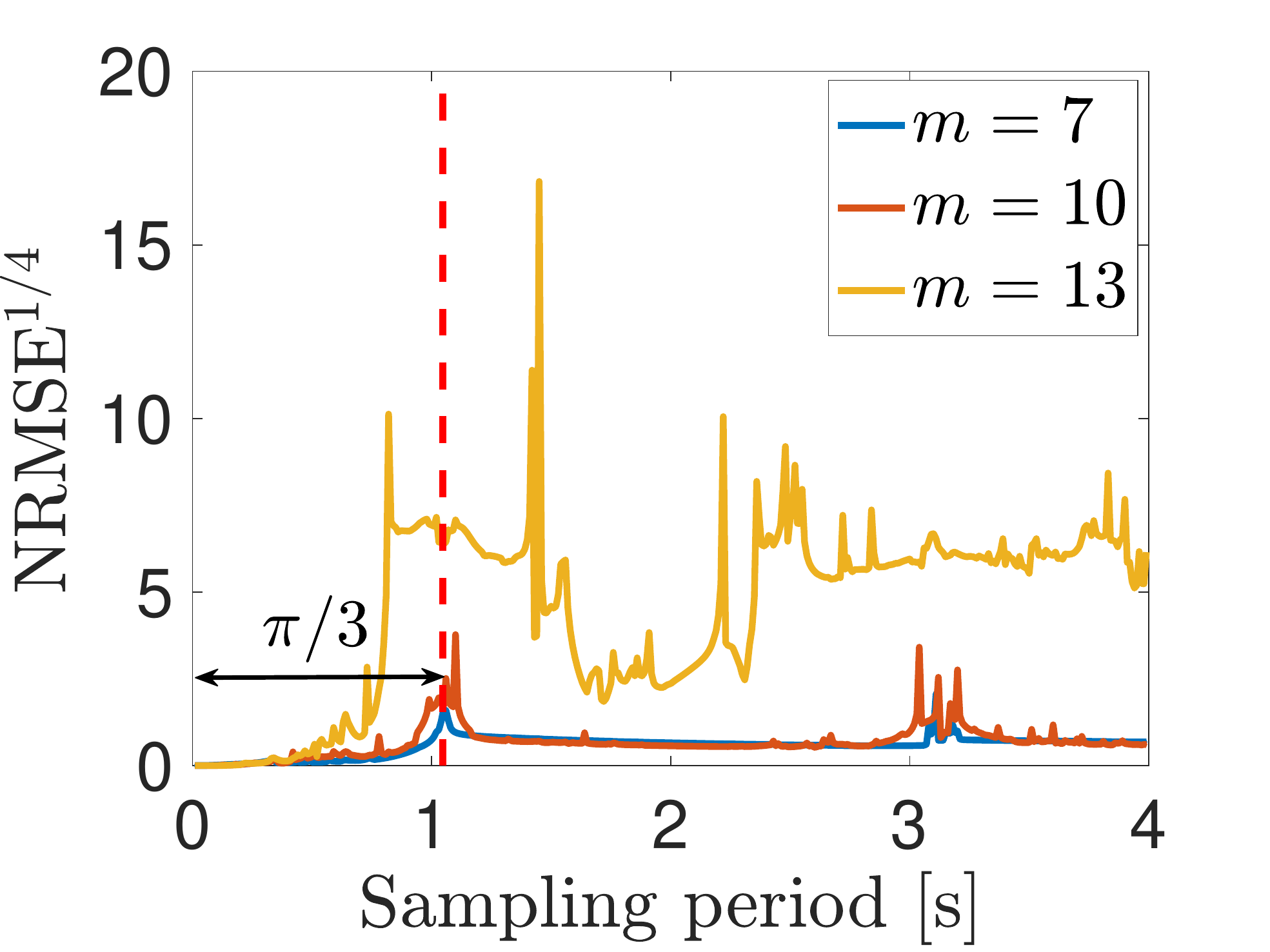}}
		\subfigure[State prediction ($T = 0.5$ s)]{
			\label{Fig6.sub2}
			\includegraphics[width=.225\textwidth]{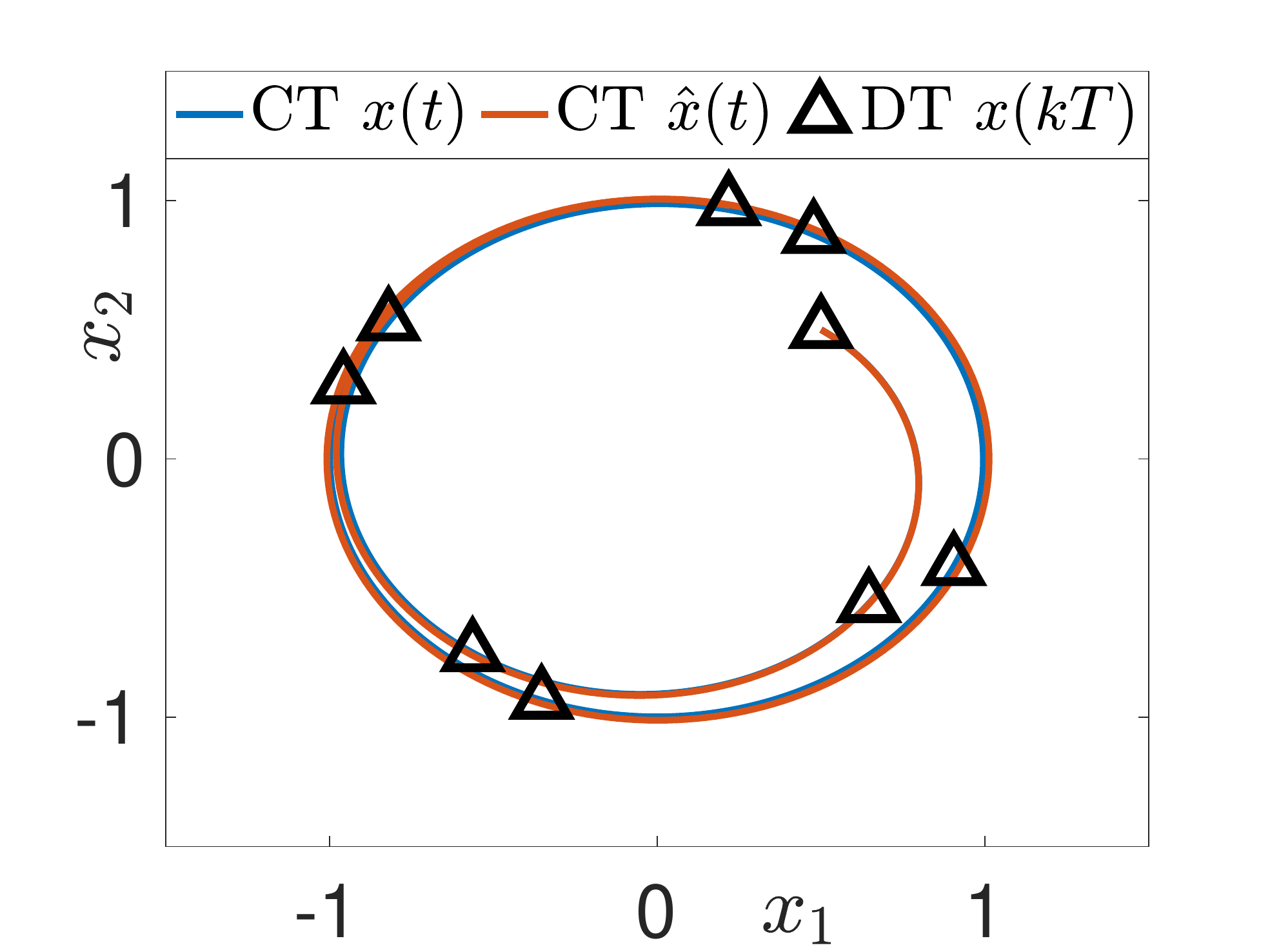}}
		\subfigure[State prediction ($T = 1.1$ s)]{
			\label{Fig6.sub3}
			\includegraphics[width=.225\textwidth]{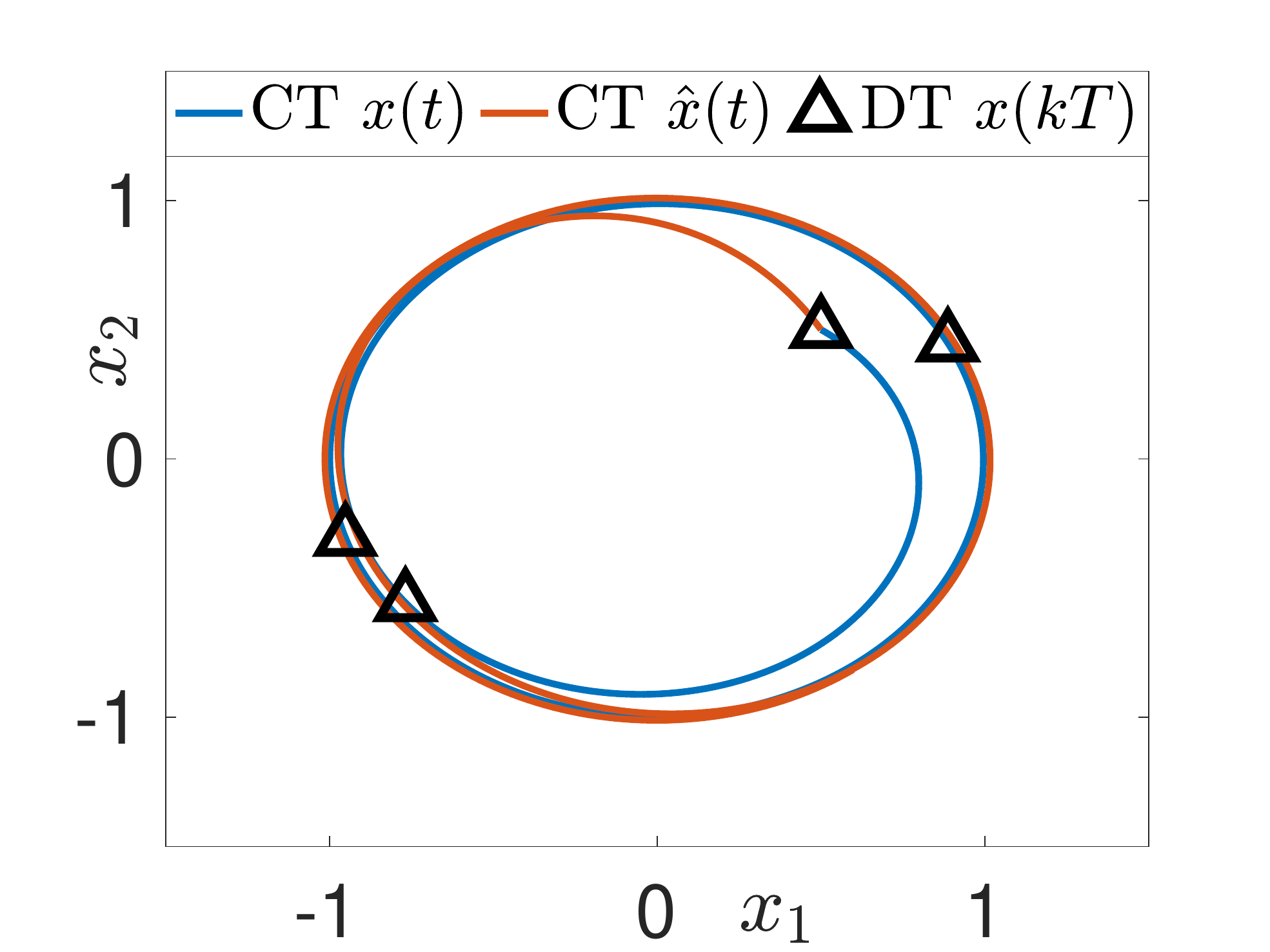}
		}
		\subfigure[State prediction ($T = 2.8$ s)]{
			\label{Fig6.sub4}
			\includegraphics[width=.225\textwidth]{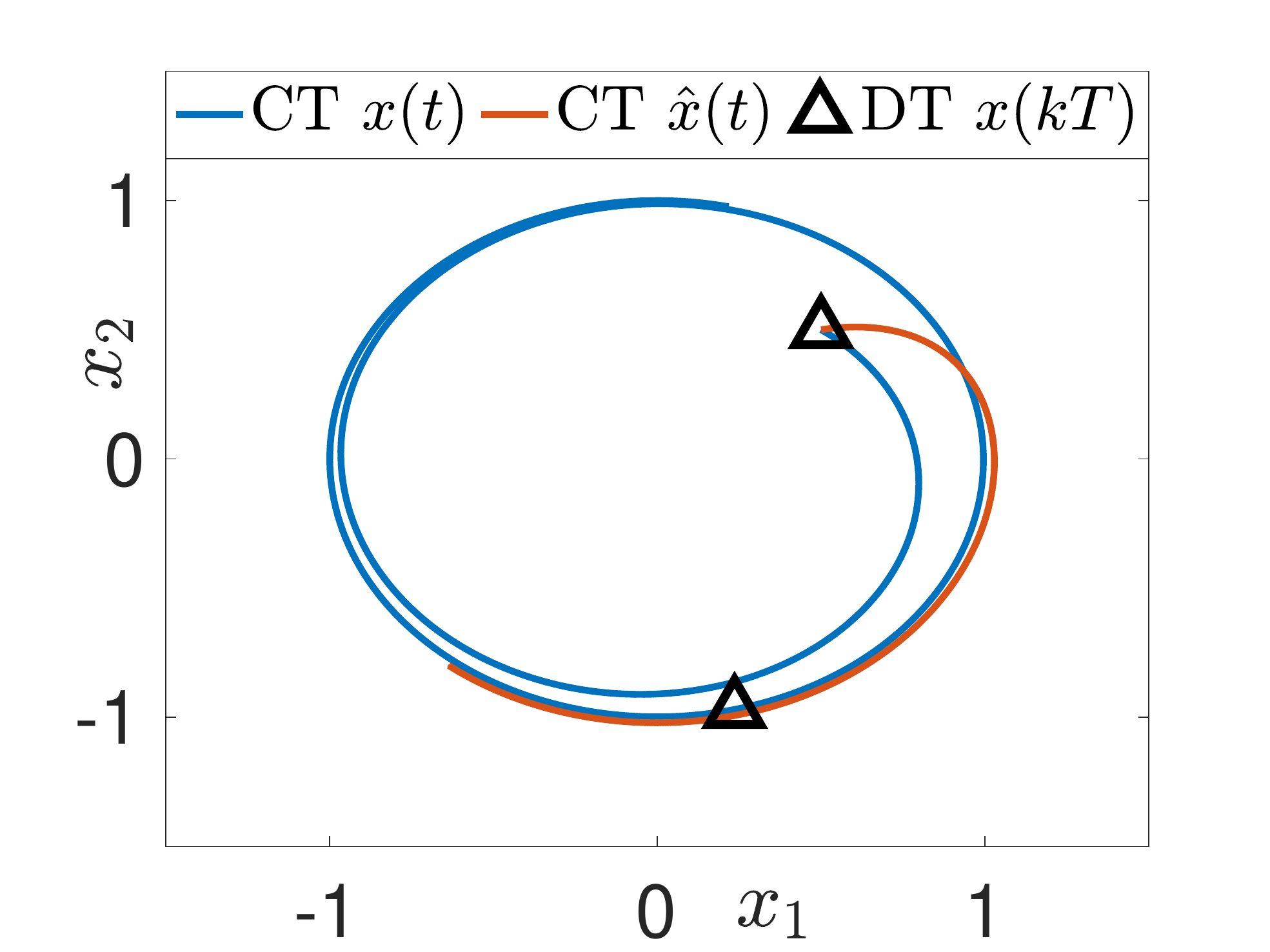}
		}
		\caption{Identification errors with growing sampling periods (a), and prediction by vector field identified from data with sampling period $T=0.5$ s (b), $T=1.1$ s (c), and $T=2.8$ s (d) of nonlinear system with limit cycle \eqref{sys3}. The parameter of the observable space is $m=13$ for prediction.}
		\label{fig6}
	\end{figure}

\emph{e) The nonlinear system with nonpolynomial vector field which admits a fixed point at $[0,0]$:}
\begin{equation}\label{sys5}
\begin{aligned}
\dot{x}_1 &= -x_1+4\frac{x_2}{1+x_2^2},\\
\dot{x}_2 &= -x_2-4\frac{x_1}{1+x_2^2}.
\end{aligned} 
\end{equation}

We generate 1000 trajectories, 30 snapshots at times for each trajectory and random initial condition $\bx(t_0)\in [-1,1]^2$. Since the vector field is nonpolynomial, the basis function of the observable space are set by the monomial functions $g(\bx) = x^{s_1}_1\ldots x_n^{s_n}|(s_1,\ldots ,s_n) \in \bN^n, s_1+\ldots +s_n\le m$ and \begin{equation*}
h(\bx) = \frac{x_l}{1+x_k^p},~k,l\in\{1,2 \},~p\in\{1,\ldots,P \}.
\end{equation*}

Four observable spaces, i.e., $\{m=1,P=2\},\{m=1,P=3\},\{m=2,P=2\},\{m=2,P=3\}$, are chosen to identify the nonlinear system \eqref{sys5} with growing sampling period. The NRMSE is also computed for different sampling period.

Since this dynamical system has the fixed point $x^* = [0,0]$, the so-called principal Koopman eigenvalues $\lambda_i$ are the eigenvalues of the Jacobian matrix of the vector field $f$ at $[0,0]$ \cite{mauroy2016global}. Therefore, the principal eigenvalues of the generator are the eigenvalues of the matrix: $$
J = \left(\begin{matrix}
-1&4\\-4&-1
\end{matrix}\right).
$$
It leads to the principal eigenvalues being $-1\pm4i$, and the critical sampling period is $\pi/4$ s. Other Koopman eigenvalues can also be obtained from these principal eigenvalues. The bound of sampling frequency will be integer multiples of $8$ rad/s if we use those eigenvalues to analyze the bound of sampling period. Therefore, the principal eigenvalues are used to analyze the critical sampling period.

The critical sampling period $\pi/4$ is denoted as the red dashed line in Fig. \ref{Fig10.sub1}. It shows that the NRMSE of the nonlinear system \eqref{sys5} appears to a peak when the sampling period grows close to the red dashed line in every observable space we choose, which is in agreement with Theorem \ref{th7}. 

\begin{figure}[thpb]
	\subfigure[Identification error]{
		\label{Fig10.sub1}
		\includegraphics[width=.225\textwidth]{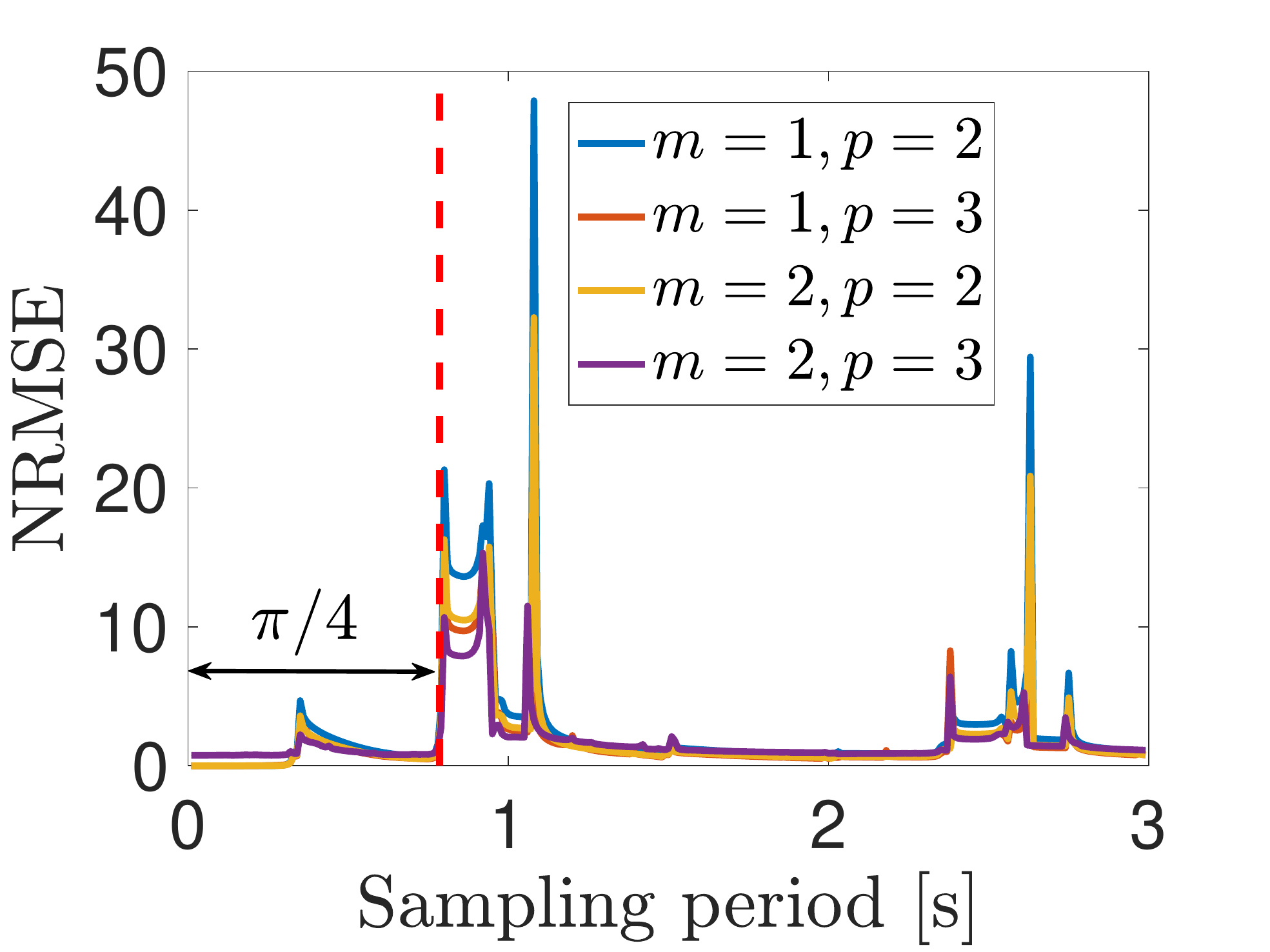}}
	\subfigure[State prediction ($T = 0.3$ s)]{
		\label{Fig10.sub2}
		\includegraphics[width=.225\textwidth]{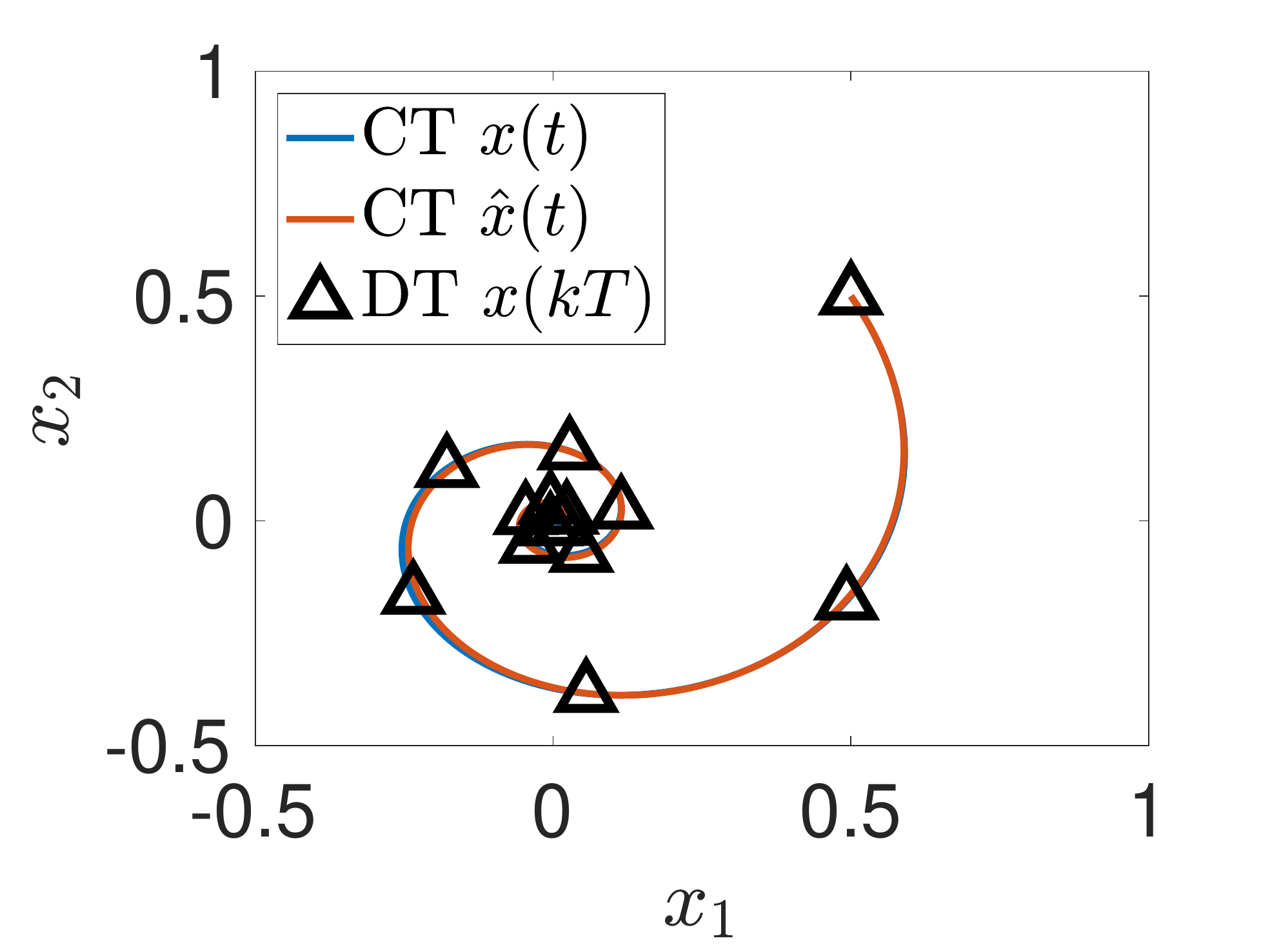}}
	\subfigure[State prediction ($T = 0.9$ s)]{
		\label{Fig10.sub3}
		\includegraphics[width=.225\textwidth]{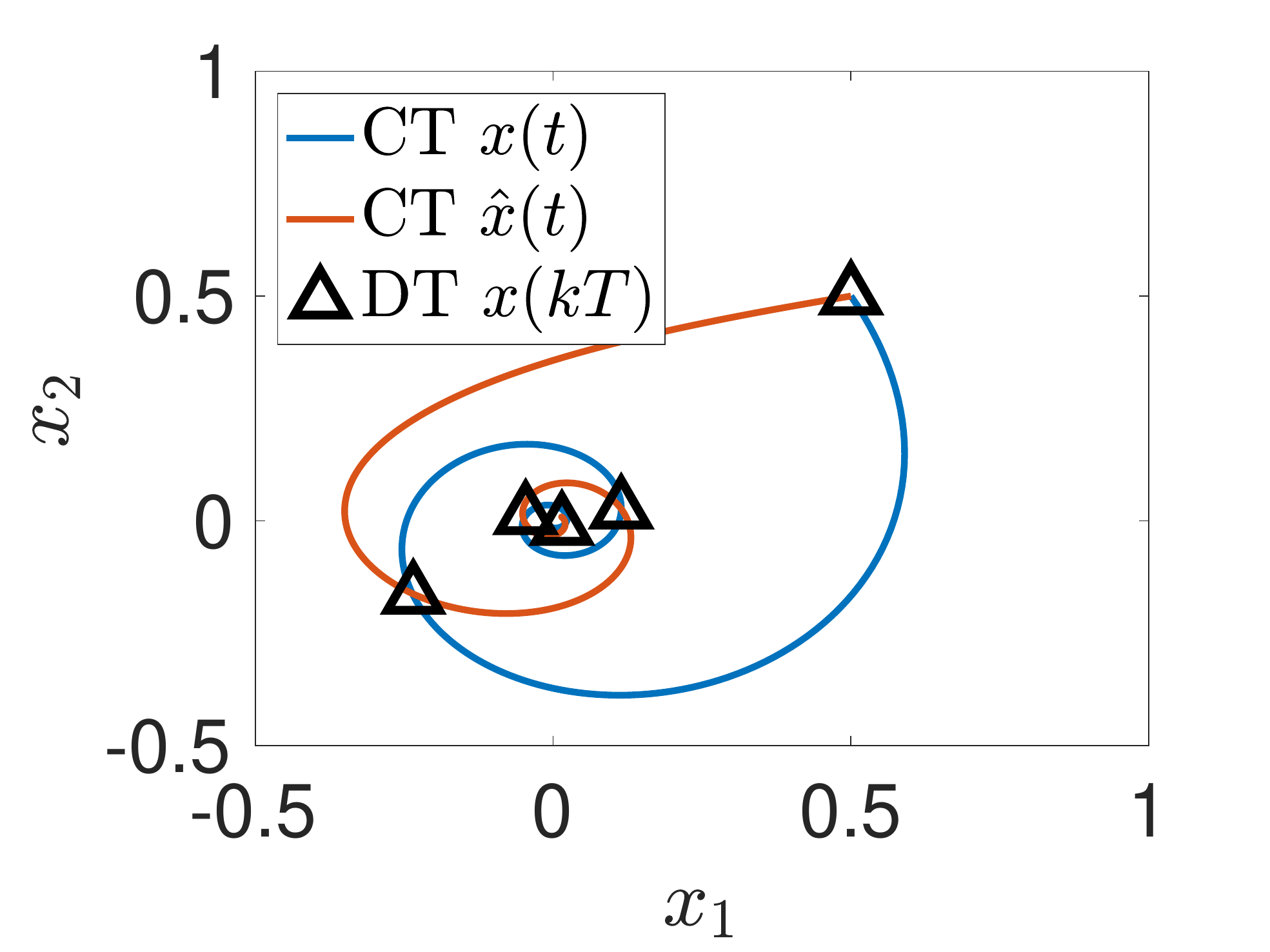}
	}
	\subfigure[State prediction ($T = 2.1$ s)]{
		\label{Fig10.sub4}
		\includegraphics[width=.225\textwidth]{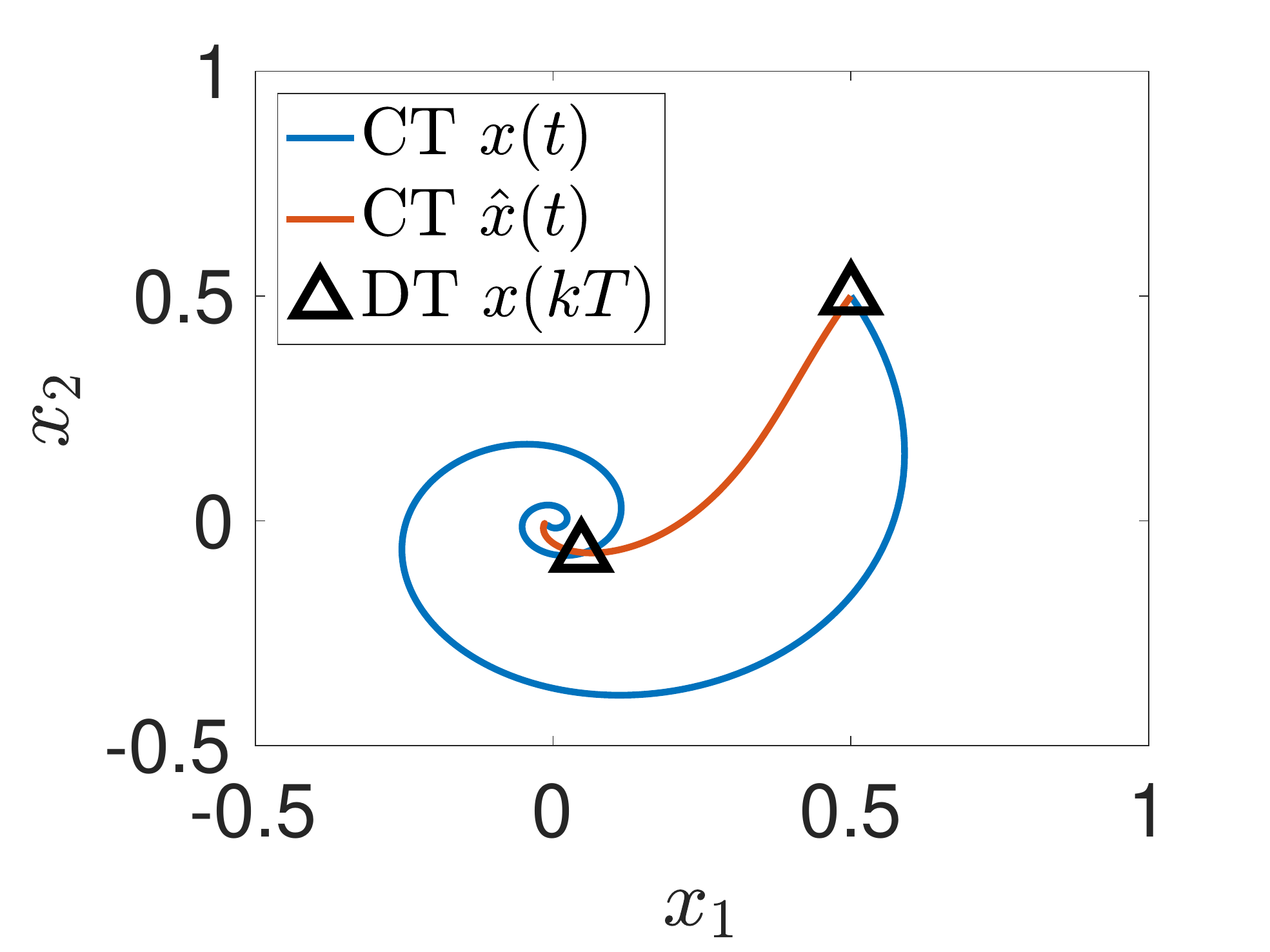}
	}
	\caption{Identification errors with growing sampling periods (a), and prediction by vector field identified from data with sampling period $T=0.3$ s (b), $T=0.9$ s (c), and $T=2.1$ s (d) of nonlinear system with nonpolynomial vector field \eqref{sys5}. The parameters of the observable space are $m=1,P=2$ for prediction.}
	\label{fig10}
\end{figure}

	It can be seen that NRMSE drops for a while then rises again when the sampling frequency exceeds the critical sampling period in Fig. \ref{Fig4.sub1}-Fig. \ref{Fig10.sub1}. To study this phenomenon, the states $\hat{\bx}(t)$ generated by the identified dynamical systems and the true states $\bx(t)$ are visualized in Fig. \ref{Fig4.sub2}-Fig. \ref{Fig4.sub4}, Fig. \ref{Fig5.sub2}-Fig. \ref{Fig5.sub4}, Fig. \ref{Fig6.sub2}-Fig. \ref{Fig6.sub4} when the sampling periods are $0.5$s, $1.1$s and $2.8$s for dynamical systems \eqref{sys1}-\eqref{sys3}, and $0.3$s, $0.9$s, $2.1$s for nonpolynomial dynamical system \eqref{sys5} in Fig. \ref{Fig10.sub2}-Fig. \ref{Fig10.sub4}. The blue lines denote the true states of time $\bx(t)$ with the initial condition $[0.5,0.5]$. The orange lines are the prediction of states $\hat{\bx}(t)$ by identified vector fields $\hat{f}$. It can be inferred that the directions of vector fields seem like rotating with the increasing of the sampling period to find the ''simplest'' trajectory matching the measurements. For example, Fig. \ref{Fig5.sub2} shows the predictions that sampled data are able to capture the correct systems. When the sampling period exceeds the critical sampling period, the direction of identified vector field $\hat{f}$ at each $\bx$ is almost opposite to the truth, and the NRMSE curve has a peak (see the prediction of the states when the sampling period is $1.1$s in Fig. \ref{Fig5.sub3}). After that, the direction of $\hat{f}$ at each $\bx$ continues rotating. Therefore, the difference between $\hat{f}$ and the truth decreases as illustrated in Fig. \ref{Fig5.sub4}.  

	We also notice that NRMSE of the dynamical system also increases slowly when $T\le T_\gamma$. Here we show that the frequency domain information of the state variable $\bx(t)$, reconstructed by the identified vector field and $\bx(0)$, almost unaffected by NRMSE of the vector field when $T\le T_\gamma$. For these four systems, discrete Fourier transform is applied to the true states $\bx(t)$ and the predicted states $\hat{\bx}(t)$ with the initial condition $\bx(0) = [0.5,0.5]$. We use a fast Fourier transform algorithm with the sampling frequency $100$ Hz and $500$ samples. The spectral error of $x_k(t)$ is computed by $\|\hat{P}_k-P_k\|_2^2$, where $\hat{P}_k$ and $P_k$ are the discrete Fourier transforms of $\hat{x}_k(t)$ and ${x_k(t)}$ respectively.  The result of the spectral error with growing sampling period can be seen in Fig. \ref{fig8}. The blue lines and orange lines denote the spectral error of states $x_1(t), x_2(t)$, and the red dashed lines denote the critical sampling period of associated dynamical system. It shows that there is little loss to reconstruct $\bx(t)$ when the sampling period is lower than its critical sampling period. However, it fails to recover $\bx(t)$ when the sampling period exceeds the bound. This implies that the critical sampling period with respect to aliasing of dynamical system may also be a criterion to avoid aliasing of a signal of time $\bx(t)$ indirectly.	
	
	%The following experiments allow to demonstrate that this lower bound $2\min_{\hF_n}\{\max|\bIm(\lambda(L|_{\hF_n}))|\}$ still holds for reconstruction of signals $\bx(t)$. For \eqref{sys3}, fast Fourier transform is applied to the correct states $\bx(t)$ and the predicted states $\hat{\bx}(t)$ with sampling frequency $100$ Hz. 
	
	\begin{figure}[!t]
			\centering
			\subfigure[Linear system]{
			\label{Fig8.sub1}
			\includegraphics[width=.48\textwidth]{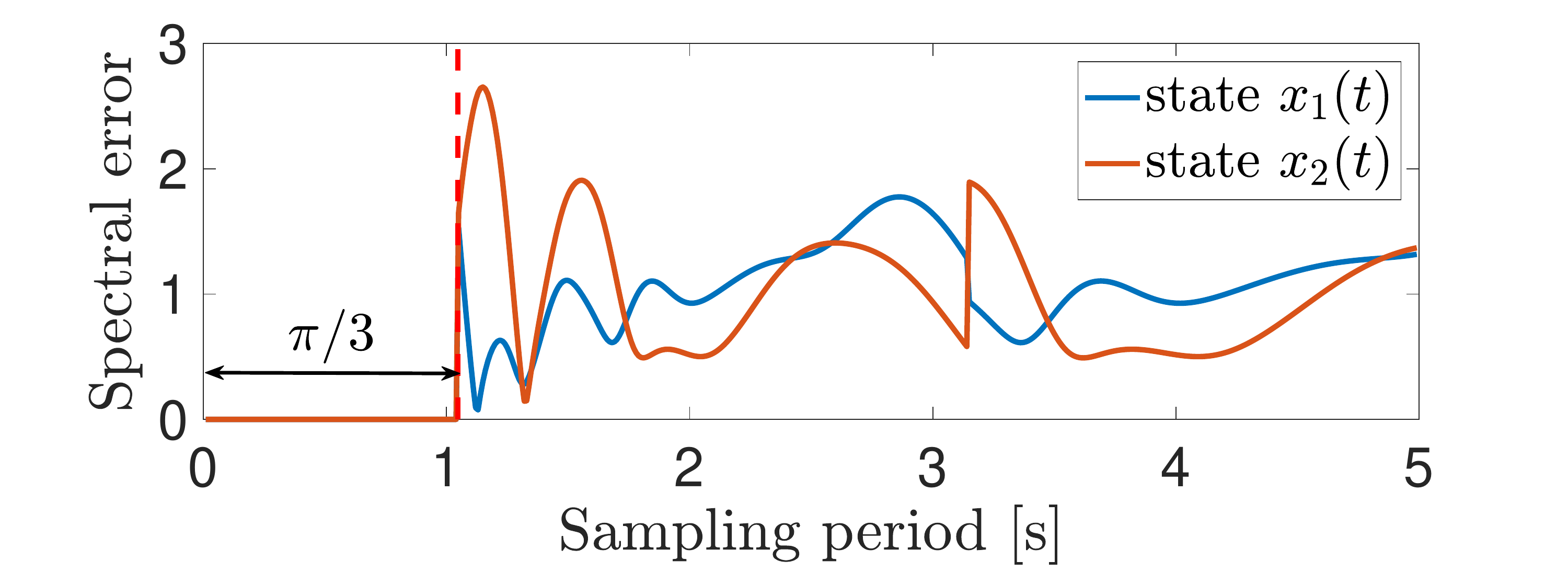}}
			\subfigure[Nonlinear system with fixed point]{
			\label{Fig8.sub2}
			\includegraphics[width=.48\textwidth]{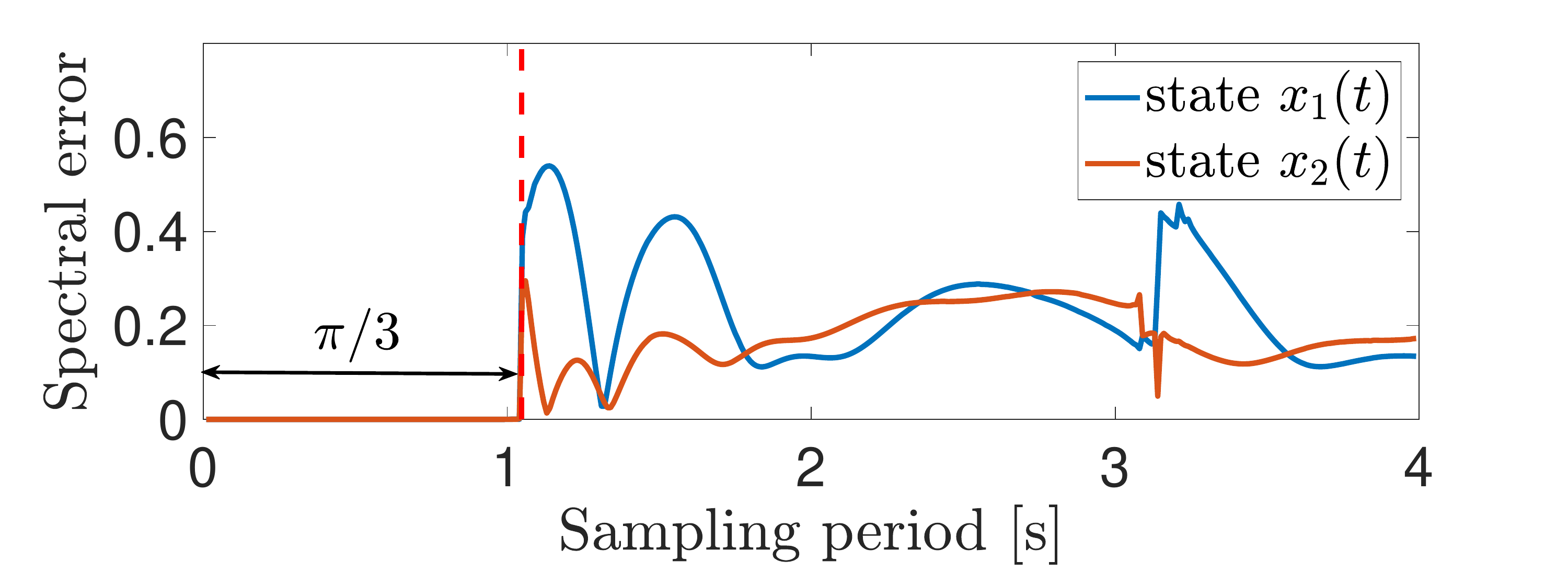}}
			\subfigure[Nonlinear system with limit cycle]{
			\label{Fig8.sub3}
			\includegraphics[width=.48\textwidth]{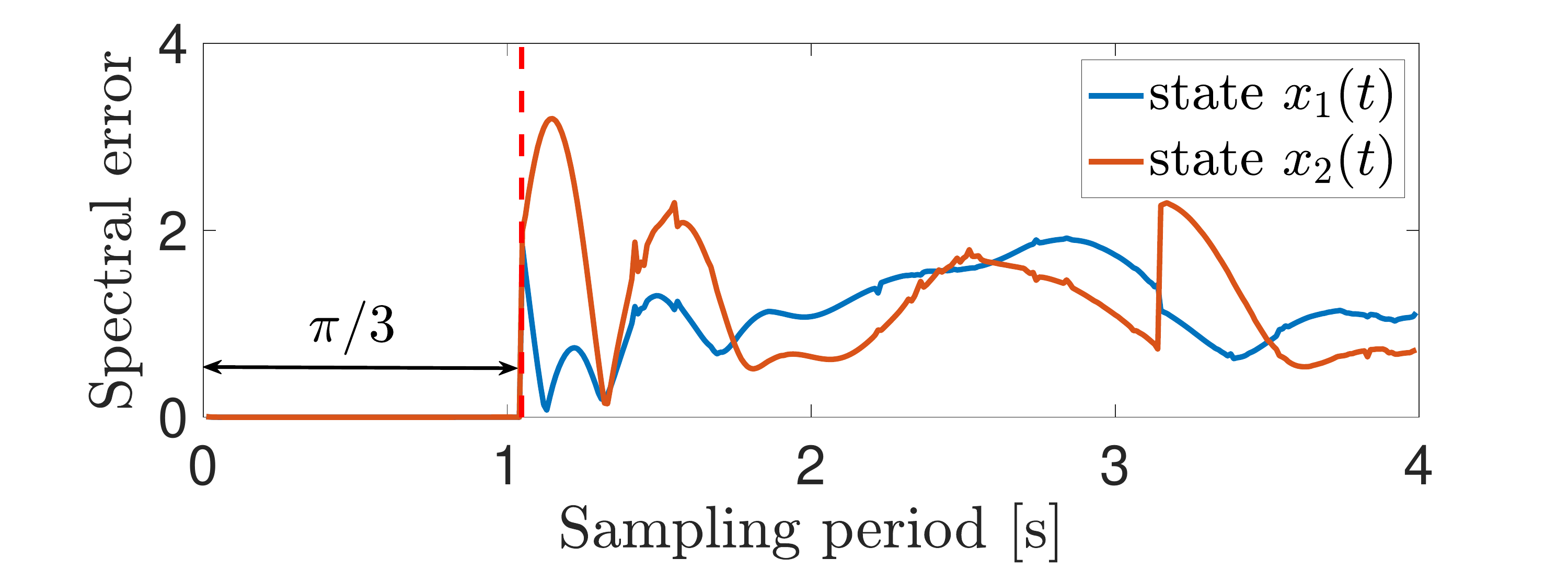}}
			\subfigure[Nonlinear system with nonpolynomial vector field]{
			\label{Fig8.sub4}
			\includegraphics[width=.48\textwidth]{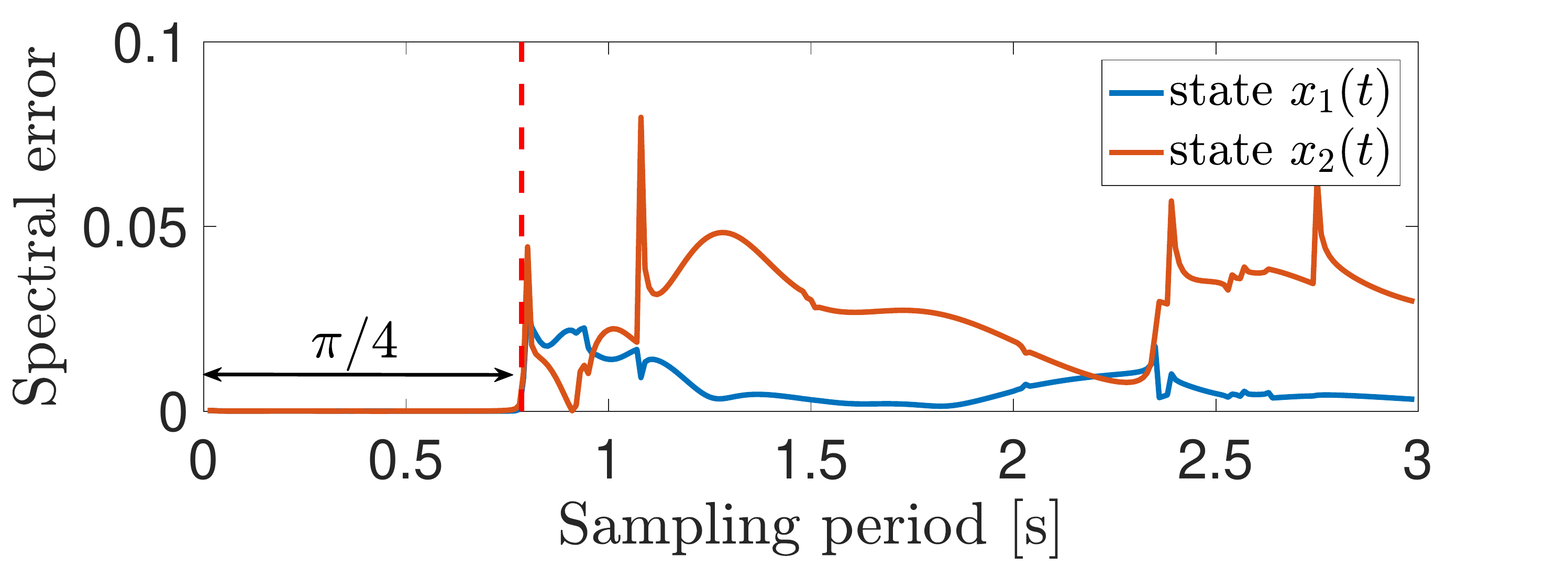}}
			\caption{Spectral error of predicted states $x_1(t)$ and $x_2(t)$ with growing sampling periods of linear system (a), nonlinear system with fixed point (b), limit cycle (c), and nonlinear system whose vector field is nonpolynomial (d). The parameter of observable space is $m=1$ for the linear system, $m=13$ for nonlinear systems, and $m=1,P=2$ for the dynamical system with nonpolynomial vector field.}
			\label{fig8}
	\end{figure}

	%----------------------------------------------------------------------------
	\section{Conclusion}\label{sec:con}
	%----------------------------------------------------------------------------
	This paper considers the system aliasing of nonlinear dynamical systems caused by low sampling frequency. It proposes necessary and sufficient conditions on the critical sampling period that avoids system aliasing. When the Koopman operator of the dynamical system has enough valid eigenfunctions, the critical sampling period can be analytically computed from the imaginary part of the corresponding Koopman eigenvalues. When it does not have enough valid eigenfunctions, the critical sampling period is analyzed by imaginary part of the corresponding spectrum, including continuous or residual spectrum, of the infinitesimal generator. Through examples on different types of dynamics, such as those with equilibria and limit cycles, the critical sampling period offers a fundamental limit of performance for any system identification algorithms, which could be essential for experimental design.
	
	\section{Acknowledgement}
	We wish to thank Prof. Yiwei Zhang (Huazhong University of Science and Technology) and Wenchao Li (Peking University) for discussion.

	\bibliographystyle{IEEEtran}
	\bibliography{IEEEabrv,cite}

% Generated by IEEEtran.bst, version: 1.14 (2015/08/26)
\begin{thebibliography}{10}
\providecommand{\url}[1]{#1}
\csname url@samestyle\endcsname
\providecommand{\newblock}{\relax}
\providecommand{\bibinfo}[2]{#2}
\providecommand{\BIBentrySTDinterwordspacing}{\spaceskip=0pt\relax}
\providecommand{\BIBentryALTinterwordstretchfactor}{4}
\providecommand{\BIBentryALTinterwordspacing}{\spaceskip=\fontdimen2\font plus
\BIBentryALTinterwordstretchfactor\fontdimen3\font minus
  \fontdimen4\font\relax}
\providecommand{\BIBforeignlanguage}[2]{{%
\expandafter\ifx\csname l@#1\endcsname\relax
\typeout{** WARNING: IEEEtran.bst: No hyphenation pattern has been}%
\typeout{** loaded for the language `#1'. Using the pattern for}%
\typeout{** the default language instead.}%
\else
\language=\csname l@#1\endcsname
\fi
#2}}
\providecommand{\BIBdecl}{\relax}
\BIBdecl

\bibitem{ljung2021modeling}
L.~Ljung, T.~Glad, and A.~Hansson, \emph{Modeling and identification of dynamic
  systems}.\hskip 1em plus 0.5em minus 0.4em\relax Studentlitteratur, 2021.

\bibitem{Yue2020a}
Z.~Yue, J.~Thunberg, L.~Ljung, Y.~Yuan, and J.~Gon{\c{c}}alves, ``System
  aliasing in dynamic network reconstruction:issues on low sampling
  frequencies,'' \emph{IEEE Transactions on Automatic Control}, vol.~66,
  no.~12, pp. 5788--5801, 2021.

\bibitem{Koopman1931hamiltonian}
B.~O. Koopman, ``{Hamiltonian systems and transformation in Hilbert space},''
  \emph{Proceedings of the National Academy of Sciences of the United States of
  America}, vol.~17, no.~5, p. 315, 1931.

\bibitem{mezic2005spectral}
I.~Mezi{\'c}, ``Spectral properties of dynamical systems, model reduction and
  decompositions,'' \emph{Nonlinear Dynamics}, vol.~41, no.~1, pp. 309--325,
  2005.

\bibitem{mauroy2019koopman}
A.~Mauroy and J.~Gon{\c{c}}alves, ``Koopman-based lifting techniques for
  nonlinear systems identification,'' \emph{IEEE Transactions on Automatic
  Control}, vol.~65, no.~6, pp. 2550--2565, 2019.

\bibitem{garnier2018contsid}
H.~Garnier and M.~Gilson, ``Contsid: a matlab toolbox for standard and advanced
  identification of black-box continuous-time models,''
  \emph{IFAC-PapersOnLine}, vol.~51, no.~15, pp. 688--693, 2018.

\bibitem{sinha1982identification}
N.~Sinha and G.~Lastman, ``Identification of continuous-time multivariable
  systems from sampled data,'' \emph{International Journal of Control},
  vol.~35, no.~1, pp. 117--126, 1982.

\bibitem{unbehauen1998review}
H.~Unbehauen and G.~Rao, ``A review of identification in continuous-time
  systems,'' \emph{Annual reviews in Control}, vol.~22, pp. 145--171, 1998.

\bibitem{garnier2008identification}
H.~Garnier and L.~Wang, \emph{Identification of continuous-time models from
  sampled data}.\hskip 1em plus 0.5em minus 0.4em\relax Springer, 2008, vol.
  202.

\bibitem{garnier2015direct}
H.~Garnier, ``Direct continuous-time approaches to system identification.
  overview and benefits for practical applications,'' \emph{European Journal of
  control}, vol.~24, pp. 50--62, 2015.

\bibitem{subrahmanyam2019identification}
A.~Subrahmanyam and G.~P. Rao, \emph{Identification of continuous-time systems:
  linear and robust parameter estimation}.\hskip 1em plus 0.5em minus
  0.4em\relax CRC Press, 2019.

\bibitem{ljung2003initialisation}
L.~Ljung, ``Initialisation aspects for subspace and output-error identification
  methods,'' in \emph{2003 European Control Conference (ECC)}.\hskip 1em plus
  0.5em minus 0.4em\relax IEEE, 2003, pp. 773--778.

\bibitem{shannon1949communication}
C.~E. Shannon, ``Communication in the presence of noise,'' \emph{Proceedings of
  the IRE}, vol.~37, no.~1, pp. 10--21, 1949.

\bibitem{glover1974parametrizations}
K.~Glover and J.~Willems, ``Parametrizations of linear dynamical systems:
  Canonical forms and identifiability,'' \emph{IEEE Transactions on Automatic
  Control}, vol.~19, no.~6, pp. 640--646, 1974.

\bibitem{yuan2016identification}
Y.~Yuan, Y.~Nakahira, and C.~J. Tomlin, ``On identification of parameterized
  switched linear systems,'' in \emph{2016 35th Chinese Control Conference
  (CCC)}.\hskip 1em plus 0.5em minus 0.4em\relax IEEE, 2016, pp. 2189--2195.

\bibitem{lasota2013chaos}
A.~Lasota and M.~C. Mackey, \emph{Chaos, fractals, and noise: stochastic
  aspects of dynamics}.\hskip 1em plus 0.5em minus 0.4em\relax Springer Science
  \& Business Media, 2013, vol.~97.

\bibitem{mauroy2020koopman}
A.~Mauroy, I.~Mezi{\'c}, and Y.~Susuki, \emph{The Koopman Operator in Systems
  and Control: Concepts, Methodologies, and Applications}.\hskip 1em plus 0.5em
  minus 0.4em\relax Springer Nature, 2020, vol. 484.

\bibitem{mauroy2016global}
A.~Mauroy and I.~Mezi{\'c}, ``Global stability analysis using the
  eigenfunctions of the koopman operator,'' \emph{IEEE Transactions on
  Automatic Control}, vol.~61, no.~11, pp. 3356--3369, 2016.

\bibitem{mezic2020spectrum}
I.~Mezi{\'c}, ``Spectrum of the koopman operator, spectral expansions in
  functional spaces, and state-space geometry,'' \emph{Journal of Nonlinear
  Science}, vol.~30, no.~5, pp. 2091--2145, 2020.

\bibitem{hille1996functional}
E.~Hille and R.~S. Phillips, \emph{Functional analysis and semi-groups}.\hskip
  1em plus 0.5em minus 0.4em\relax American Mathematical Soc., 1996, vol.~31.

\bibitem{krabbe1956logarithm}
G.~Krabbe, ``On the logarithm of a uniformly bounded operator,''
  \emph{Transactions of the American Mathematical Society}, vol.~81, no.~1, pp.
  155--166, 1956.

\bibitem{korda2018linear}
M.~Korda and I.~Mezi{\'c}, ``Linear predictors for nonlinear dynamical systems:
  Koopman operator meets model predictive control,'' \emph{Automatica},
  vol.~93, pp. 149--160, 2018.

\end{thebibliography}

		\begin{IEEEbiography}[{\includegraphics[width=1in,height=1.25in,clip,keepaspectratio]{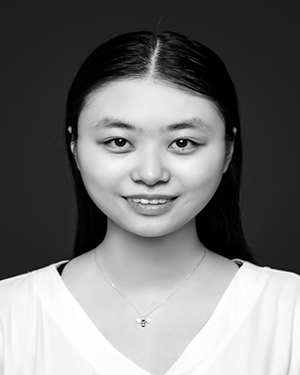}}]
		{Zhexuan Zeng} is currently a Ph.D. student at Huazhong University of Science and Technology under the supervision of Prof. Ye Yuan. Her research interests
		include system identification, nonlinear systems, and applications of operator theoretic methods.
	\end{IEEEbiography}
	
		\begin{IEEEbiography}[{\includegraphics[width=1in,height=1.25in,clip,keepaspectratio]{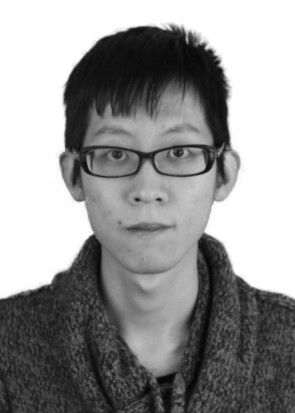}}]
		{Zuogong Yue} is currently an Assistant Professor at Huazhong University of Science and Technology. He received the B.Sc. degree from Zhejiang University, China, in 2011; the M.Phil. degree from the Hong Kong University of Science and Technology, China, in 2013; and the Ph.D. degree from the University of Luxembourg, Luxembourg, in 2018. He worked as a postdoctoral fellow at the Luxembourg Centre for Systems Biomedicine in 2018 and at the University of New South Wales in 2019-2021. His research interests include system/network identification and control, signal processing and learning.
		\end{IEEEbiography}

	\begin{IEEEbiography}[{\includegraphics[width=1in,height=1.25in,clip,keepaspectratio]{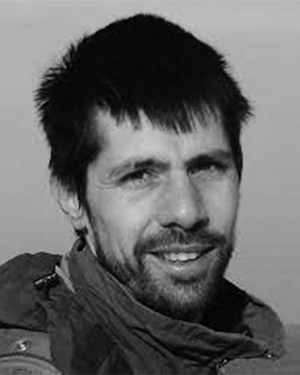}}]
			{Alexandre Mauroy} received the M.S. degree in aerospace engineering and the Ph.D. degree in applied mathematics from the University of Li `ege, Li `ege, Belgium, in 2007 and 2011, respectively. Prior to joining the University of Namur, he was a postdoctoral researcher with the University of California Santa Barbara from 2011 to 2013, the University of Li `ege from 2013 to 2015, and the University of Luxembourg in 2016. He is currently an Assistant Professor with the Department of Mathematics and the Namur Institute for Complex Systems, University of Namur, Belgium. His research interests include synchronization in complex networks, network identification, and applications of operator theoretic methods to dynamical systems and control.
	\end{IEEEbiography}
	
	\begin{IEEEbiography}[{\includegraphics[width=1in,height=1.25in,clip,keepaspectratio]{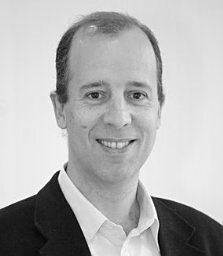}}]
		{Jorge Gon\c{c}alves} is a Professor at the Luxembourg Centre for Systems
		Biomedicine, University of Luxembourg and a Principal Research Associate
		at the Department of Plant Sciences, University of Cambridge. He received
		his Licenciatura (5-year S.B.) degree from the University of Porto,
		Portugal, and the M.S. and Ph.D. degrees from the Massachusetts Institute
		of Technology, Cambridge, MA, all in Electrical Engineering and Computer
		Science, in 1993, 1995, and 2000, respectively. He then held two
		postdoctoral positions, first at the Massachusetts Institute of
		Technology for seven months, and from 2001 to 2004 at the California
		Institute of Technology with the Control and Dynamical Systems Division.
		At the Information Engineering Division of the Department of Engineering,
		University of Cambridge he was a Lecturer from 2004 until 2012, a Reader
		from 2012 until 2014, and from 2014 until 2019 he was a Principal
		Research Associate. From 2005 until 2014 he was a Fellow of Pembroke
		College, University of Cambridge. Since
		2019 he is a Principal Research Associate at the Department of Plant
		Sciences, University of Cambridge, and since 2013 he is a Professor at
		the Luxembourg Centre for Systems Biomedicine, University of Luxembourg.
	\end{IEEEbiography}
	
	\begin{IEEEbiography}[{\includegraphics[width=1in,height=1.25in,clip,keepaspectratio]{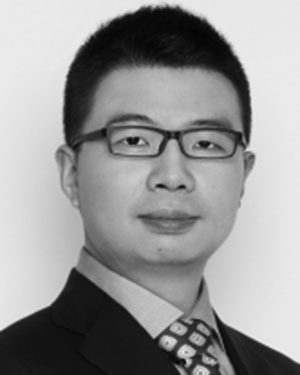}}]
		{Ye Yuan} received the B.Eng. degree (Valedictorian) from the Department
		of Automation, Shanghai Jiao Tong University, Shanghai, China, in
		September 2008, and the M.Phil. and Ph.D. degrees in control theory from
		the Department of Engineering, University of Cambridge, Cambridge, U.K.,
		in October 2009 and February 2012, respectively.  He is currently a Full
		Professor with the Huazhong University of Science and Technology, Wuhan,
		China. He was a Postdoctoral Researcher at UC Berkeley, a Junior Research
		Fellow at Darwin College, University of Cambridge. His research interests
		include system identification and control with applications to cyber-physical systems.
	\end{IEEEbiography}

\end{document}